\documentclass{revtex4}
\usepackage{amsmath,amsthm}
\usepackage{rotating}
\usepackage{multirow}
\usepackage{graphicx}

\usepackage{amsfonts}
\newtheorem{theorem}{Theorem}[section]
\newtheorem{lemma}[theorem]{Lemma}
\newtheorem{proposition}[theorem]{Proposition}

\theoremstyle{remark}
\newtheorem{remark}[theorem]{Remark}

\theoremstyle{definition}
\newtheorem{definition}[theorem]{Definition}

\theoremstyle{example}
\newtheorem{example}[theorem]{Example}

\theoremstyle{notation}
\newtheorem{notation}[theorem]{Notation}
\newcommand{\bra}[1]{\langle#1|}
\newcommand{\ket}[1]{|#1\rangle}

\begin{document}

\title{Equivalence classes of coherent projectors in a Hilbert space with prime dimension: $Q$ functions and their Gini index}            
\author{A. Vourdas}
\affiliation{Department of Computer Science,\\
University of Bradford, \\
Bradford BD7 1DP, United Kingdom\\a.vourdas@bradford.ac.uk}

\begin{abstract}

Coherent subspaces spanned by a finite number of coherent states are introduced, in a quantum system with Hilbert space that has odd prime dimension $d$.
The set of all coherent subspaces is partitioned into equivalence classes, with $d^2$ subspaces in each class.
The corresponding coherent projectors within an equivalence class, have the `closure under displacements property' and also resolve the identity.
Different equivalence classes provide different granularisation of the Hilbert space, and they form a partial order `coarser' (and `finer').
In the case of a two-dimensional coherent subspace spanned by two coherent states, the corresponding projector (of rank $2$)  is different than the sum of the two projectors to the subspaces related to each of the 
two coherent states. We quantify this with `non-addditivity operators' which are a measure of  quantum interference in phase space, and also of the non-commutativity of the projectors.
Generalized $Q$ and $P$ functions of density matrices, which are based on coherent projectors in a given equivalence class, are introduced.
Analogues of the Lorenz values and the Gini index (which are popular quantities in Mathematical Economics) are used here to quantify 
the inequality in the distribution of the $Q$ function of a quantum state,
within the granular structure of the Hilbert space. A comparison is made between Lorenz values and the Gini index for the cases of coarse and also fine granularisation of the Hilbert space.
 Lorenz values require an ordering of the $d^2$ values of the $Q$ function of a density matrix, and this leads 
to the  ranking permutation of a density matrix, and to comonotonic density matrices (which have the same ranking permutation).
The Lorenz values are a superadditive function and the Gini index is a subadditive function (they are both  additive quantities for comonotonic density matrices).
Various examples demonstrate these ideas.
\end{abstract}
\maketitle

\section{Introduction}
Coherent states in quantum systems with infinite-dimensional Hilbert space have been studied for a long time\cite{A1,A2}. 
A generalisation of this is the concept of coherent subspaces that are spanned by a finite number of coherent states \cite{V2016}.
The projectors to the coherent subspaces (which we call coherent projectors) have properties analogous to coherent states 
(closure under displacement transformations, resolution of the identity, etc).

In this paper we consider quantum systems with variables in ${\mathbb Z}_d$ (the integers modulo $d$) and $d$-dimensional Hilbert space $H_d$.
 Coherent states in this context are defined in a similar way to the the case of infinite-dimensional Hilbert space,
 by acting with the displacement operators on some fiducial vector.
 The number of coherent states is in this case $d^2$.

 In this context we introduce coherent subspaces spanned by a finite number of coherent states, 
and the corresponding coherent projectors.
We then show that the set of all coherent subspaces is partitioned into equivalence classes, with $d^2$ subspaces in each class.
The corresponding coherent projectors within an equivalence class have the closure under displacements property and also they resolve the identity. In this sense each equivalence class provides a covering of the Hilbert space $H_d$.
Coherent subspaces form a granular structure of the Hilbert space $H_d$,  
and different equivalence classes form different granularisations which are partially ordered with a partial order `coarser' (and `finer').

A technical point here is that in order to calculate the number of equivalence classes (and also some other results) it is convenient to take 
the dimension $d$ to be a prime number (other than $2$). ${\mathbb Z}_d$ is in general a ring but for prime $d$ it is a field, and this simplifies many calculations.
Systems with prime dimension $d$ are known to have stronger properties, e.g., in relation to mutually unbiased bases (e.g., \cite{V1,F1,F2}), but this is a
different area unrelated to the present paper. Here the results for the granular structure of the Hilbert space $H_d$ and the equivalence classes, are simpler for the case of 
odd prime $d$. In future work it is of course interesting to generalise this structure to an arbitrary $d$.

In the case of a two-dimensional coherent subspace spanned by two coherent states, the corresponding projector is different than the sum of the two projectors to the subspaces related to each of the 
two coherent states. In order to quantify this we introduce `non-additivity operators' and study their properties.
Their sum within an equivalence class of projectors is zero, because constructive interference is equal to destructive interference.

Using the coherent projectors within a given equivalence class, we define generalised $Q$ and $P$ representations.
A Hermitian operator $\theta$ is described with $d^2$ values of the $Q$ function, and also with $d^2$ values of the $P$ function.
We then use analogues of the Lorenz values and the Gini index (which are popular quantities in Mathematical Economics \cite{Gini}) to quantify 
the inequality in the distribution of the $Q$ function of a quantum state, within the granular structure of the Hilbert space.
This is a novel way of characterising quantum states.
A comparison of  the Lorenz values and the Gini index with coarse versus  fine granular structures, is made in propositions \ref{pro250}, \ref{pro251}.

In order to calculate the Lorenz values we order  the $d^2$ values of the $Q$ function of a density matrix (from the `poorest' to the `richest'), and this leads 
to the concept of the ranking permutation, which indicates the location of the density matrix in the Hilbert space.
We have used this concept previously in refs \cite{C1,C2} in connection with  Choquet integrals.
We have also studied comonotonic density matrices which have the same ranking permutation, and in this sense they `live' in the same part of the Hilbert space and they are `physically similar'.
We note that comonotonicity and Choquet integrals  have been used extensively (with scalar quantities) in Mathematical Economics, Decision Theory, Artificial Inteligence, etc (e.g., \cite{D1,D2,D3,D4,D5}).
We have used them in a quantum context with matrices \cite{C1,C2}.

In the present paper we use the ranking permutation of a density matrix and comonotonicity, in connection with the Lorenz values and the Gini index for $Q$ functions.
For comonotonic density matrices both the Lorenz values and the Gini index are additive quantities (propositions \ref{pro66},\ref{pro67}).
 
 Work on coherence in the context of category theory and logic (e.g., \cite{L1,L2}) might also be linked to the present work, and further work is needed in this direction. 
 
\subsection{Main results}
 \begin{itemize}
\item
We introduce multi-dimensional coherent subspaces (spanned by a finite number of coherent states).
The set of all coherent subspaces is partitioned into equivalence classes, with $d^2$ subspaces in each class (proposition \ref{pro1}).
The coherent projectors in an equivalence class have the closure property under displacements in Eq.~(\ref{33}), and the resolution of the identity in Eq.~(\ref{2A}). 
\item
The projectors related to two-dimensional coherent subspaces spanned by two coherent states obey the `non-additivity' inequality in Eq.~(\ref{dfg}).
Motivated by this we introduce `non-additivity operators' in Eq.(\ref{rty}), which are related to quantum interference in phase space.
Their sum within an equivalence class of projectors is zero, because constructive interference is equal to destructive interference (Eq.~(\ref{bn})).
\item
The coherent subspaces within an equivalence class provide a covering of the Hilbert space $H_d$.
The various coverings form the partial order  'coarser' (definitions \ref{def10}, \ref{def11}, and proposition \ref{pro3}).
\item
Using the coherent projectors within a given equivalence class, we define generalised $Q$ and $P$ representations and study their properties (proposition \ref{pro135}).
\item
We introduce the ranking permutation of a density matrix which can be viewed as its postcode in the Hilbert space.
Density matrices that have the same ranking permutation are comonotonic.
Comonotonicity is an equivalence relation in the set of density matrices, and partition it into equivalence classes which are convex sets (proposition \ref{pro34}).
\item
We use  Lorenz values and the Gini index to quantify the inequality in the distribution of the $Q$ function of a quantum state,
within the granular structure of the Hilbert space.
A comparison of  the Lorenz values and the Gini index in the cases of coarse and fine granularisations, is made in propositions \ref{pro250}, \ref{pro251}.
The Lorenz values are a superadditive function and the Gini index is a subadditive function.
For comonotonic density matrices, they are both  additive quantities (propositions \ref{pro66}, \ref{pro67}).
 \end{itemize}

\subsection{Contents}
In section 2 we discuss briefly quantum systems with variables in ${\mathbb Z}_d$, in order to define the notation.
We also explain in a general context, that the projector $\Pi(H_1\vee H_2)$ to the disjunction of two subspaces $H_1, H_2$ is different than the sum
$\Pi(H_1)+\Pi(H_2)$ of the projectors to these subspaces.
We quantify this with the non-additivity operator ${\mathfrak d}(H_1,H_2)$ which is intimately related to the commutator 
$[\Pi(H_1), \Pi(H_2)]$ (Eq.~(\ref{17})).

In section 3 we introduce  coherent subspaces and coherent projectors.
In section 4 we consider subsets of the phase space ${\mathbb Z}_d\times {\mathbb Z}_d$ and their dispalecements.
This is background material for section 5 which introduces various  granular structures to the phase space and the corresponding granular structures  to the Hilbert space.
Each granular structure consists of subsets of the phase space which cover the full phase space, and also of the  corresponding coherent subspaces which cover the full Hilbert space.
There are various equivalence classes of these granular structures of the phase space and of the Hilbert space, which form a partial order `coarser' (and `finer').
They are discussed in propositions \ref{pro1},\ref{pro2},\ref{pro3}.

In section 6, we prove the closure under displacements and the resolution of the identity properties of the coherent projectors within an equivalence class (proposition \ref{pro12}).
These two properties justify the use of the term `coherent'.
In section 7,  we use coherent projectors to define generalized $Q$ and $P$ functions for Hermitian operators.
In subsection 7A we introduce the $\delta Q$ functions (related to the `non-addditivity operators') as a measure of quantum interference in phase space.

In sections 8,9 we use  Lorenz values and the Gini index to quantify how spread is the quantum state in the granular structure of the Hilbert space.
In section 8A we partition the set of all density matrices into equivalence classes of comonotonic density matrices and show in proposition \ref{pro34} that they are convex sets.
In section 8B we discuss the Lorenz values, and in section 8C the Lorenz operators. The Gini index is discussed in section 9.
Various examples in section 9B demonstrate these ideas.
We conclude in section 10 with a summary of our results.

\section{Background}
\subsection{Quantum systems with variables in ${\mathbb Z}_d$ with odd prime $d$}

We consider a quantum system with variables in ${\mathbb Z}_d$, where $d$ is an odd  prime number.
In this case ${\mathbb Z}_d$ is a field.
$|X;r\rangle$ where $r\in {\mathbb Z}_d$, is an orthonormal basis which we call position states (the $X$ in this notation is not a variable, but it simply indicates position states).
Through a Fourier transform we get another orthonormal basis that we call momentum states:
\begin{equation}
|{P};r\rangle=F|{X};r\rangle;\;\;\;\;
F=\frac{1}{\sqrt d}\sum _{r,s}\omega(rs)\ket{X;r}\bra{X;s};\;\;\;\;
\omega(r)=\exp \left (i\frac{2\pi r}{d}\right ).
\end{equation}
The displacement operators $Z^\alpha, X^\beta$ in the phase space ${\mathbb Z}_d\times {\mathbb Z}_d$, are given by
\begin{eqnarray}\label{2A}
&&Z^\alpha= \sum \omega(\alpha m)|X; m\rangle\bra{X;m}=\sum |P; m+\alpha \rangle \bra{P:m}\nonumber\\
&&X^\beta =\sum |X; m+\beta\rangle \bra{X; m}= \sum \omega(-m\beta)|P; m \rangle \bra{P;m}.
\end{eqnarray}
where $\alpha$, $\beta \in {\mathbb Z}_d$. 
General displacement operators are the unitary operators
\begin{equation}
D(\alpha, \beta)=Z^\alpha X^\beta \omega(-2^{-1}\alpha \beta);\;\;\;[D(\alpha, \beta)]^{\dagger}=D(-\alpha, -\beta);\;\;\;\alpha, \beta \in {\mathbb Z}_d
\end{equation}
The $2^{-1}$ exists in ${\mathbb Z}_d$ because $d$ is an odd integer.
It is known that
\begin{equation}
D(\alpha, \beta)D(\gamma, \delta)=D(\alpha+\gamma, \beta+\delta)\omega[2^{-1}(\alpha \delta-\beta \gamma)].
\end{equation}
Also for an arbitrary operator $\theta$\cite{V1}
\begin{equation}\label{1A}
\frac{1}{d}\sum _{\alpha, \beta} D(\alpha, \beta)\frac{\theta}{{\rm Tr}\theta}[D(\alpha, \beta)]^\dagger={\bf 1};\;\;\;{\rm Tr}\theta \ne 0.
\end{equation}
We simplify the notation, by using $i=(\alpha, \beta)\in {\mathbb Z}_d\times {\mathbb Z}_d$.
Then the displacement operators are denoted as $D(i)$.

\subsection{Non-additivity of quantum probabilities and quantum interference}
We consider a quantum system with positions and momenta in ${\mathbb Z}_d$ (the integers modulo $d$), and $d$-dimensional Hilbert space $H_d$.
Let $H_1, H_2$ be  subspaces of $H_d$ and $\Pi(H_1), \Pi(H_2)$ the corresponding projectors. We define the conjunction (logical AND) and disjunction (logical OR):
\begin{eqnarray}\label{bb}
H_1\wedge H_2=H_1\cap H_2;\;\;\;\;\;H_1\vee H_2={\rm span}(H_1 \cup H_2).
\end{eqnarray}
We denote as ${\cal O}$ the zero-dimensional subspace that contains the zero vector.

We also consider two partial orders $\overset {H}\prec$  and $\overset {P}\prec$  for Hilbert subspaces and projectors, correspondingly (which are isomorphic to each other). 
 $H_1 \overset {H}\prec H_2$ if $H_1$ is a subspace of $H_2$. Also  $\Pi(H_1) \overset {P}\prec \Pi(H_2)$ if  $H_1 \overset {H}\prec H_2$.
If $\rho$ is a density matrix, the $\Pi(H_1) \overset {P}\prec \Pi(H_2)$ implies that ${\rm Tr}[\rho \Pi(H_1)]\le {\rm Tr}[\rho \Pi(H_2)]$ for the corresponding quantum probabilities.

We note that in general
\begin{eqnarray}\label{qqq}
H_1\wedge H_2={\cal O}\;\;\rightarrow\;\;\Pi(H_1\vee H_2)\ne \Pi(H_1)+\Pi(H_2),
\end{eqnarray}
and this leads to the following result for the corresponding quantum probabilities:
\begin{eqnarray}\label{aaa}
{\rm Tr}[\rho \Pi(H_1\vee H_2)]\ne {\rm Tr}[\rho\Pi(H_1)]+{\rm Tr}[\rho\Pi(H_2)].
\end{eqnarray}
In general these projectors do not commute, and consequently these quantum probabilities need to be measured using different ensembles described by the same density matrix $\rho$.

In practical calculations the projector $ \Pi(H_1\vee H_2)$ can be constructed from $n$ linearly independent vectors $v_1,...,v_n$ that span $H_1\vee H_2$. Let ${\cal U}$ the $d\times n$ matrix $(v_1,...,v_n)$ which has as columns the vectors $v_1,...v_n$ (where $n\le d$).
The projector to the space  $H_1\vee H_2$ is
\begin{eqnarray}
\Pi(H_1\vee H_2)={\cal U}({\cal U}^\dagger {\cal U})^{-1}{\cal U}^\dagger.
\end{eqnarray}
Since the vectors are linearly independent, ${\rm rank} (A)=n$, and therefore the matrix $A^\dagger A$ is invertible.

The inequality in Eq.~(\ref{qqq}) is a result of quantum interference, as the following example shows.
\begin{example}
Let $H_1, H_2$ be one-dimensional subspaces that contain the (normalized) column vectors $v_1,v_2$, correspondingly.
In this case
${\cal U}=\begin{pmatrix}
v_1&v_2\\
\end{pmatrix}$ and 
\begin{eqnarray}
\Pi(H_1\vee H_2)=
\begin{pmatrix}
v_1&v_2\\
\end{pmatrix}
\begin{pmatrix}
1&v_1^{\dagger}v_2\\
v_2^{\dagger}v_1&1\\
\end{pmatrix}^{-1}
\begin{pmatrix}
v_1^\dagger\\
v_2^\dagger\\
\end{pmatrix}
\end{eqnarray}
The matrix ${\cal U}^\dagger {\cal U}$ contains the off-diagonal interference terms $v_1^{\dagger}v_2$ and $v_2^{\dagger}v_1$.
From this follows that
\begin{eqnarray}
\Pi(H_1\vee H_2)=\frac{1}{1-|v_1^{\dagger}v_2|^2}
\left [\Pi(H_1)+\Pi(H_2)-(v_1^{\dagger}v_2)v_1v_2^{\dagger}-(v_2^{\dagger}v_1)v_2v_1^{\dagger}\right ]
\end{eqnarray}
It is seen that the interference terms $v_1^{\dagger}v_2$ and $v_2^{\dagger}v_1$ are responsible for the inequality in Eq.(\ref{qqq}).
For orthogonal vectors $v_1, v_2$, we have no interference and Eq.(\ref{qqq}) becomes equality.
\end{example}

\begin{remark}

In classical Boolean logic  the conjunction (logical AND) and disjunction (logical OR) are defined as
\begin{eqnarray}
A\wedge B=A\cap B;\;\;\;\;\;A\vee B=A \cup B,
\end{eqnarray}
where $A,B$ are subsets of a set $\Omega$.
In contrast the quantum  logical OR in Eq.(\ref{bb}) is not just the union, but it contains superpositions of states in the two spaces, which leads to quantum interference.

Related to this is that in a classical context Kolmogorov probabilities obey the additivity relation:
\begin{eqnarray}\label{bbb}
A\cap B=\emptyset\;\;\rightarrow\;\;p(A\cup B)= p(A)+p(B).
\end{eqnarray}
From this follows the weaker relation
\begin{eqnarray}\label{bba}
A\subseteq B\;\;\rightarrow\;\;p(A)\le p(B).
\end{eqnarray}
Quantum probabilities obey the analogue of Eq.~(\ref{bba})
\begin{eqnarray}
H_1 \overset {H}\prec H_2\;\;\rightarrow\;\;\Pi(H_1) \overset {P}\prec \Pi(H_2),
\end{eqnarray}
but due to quantum interference the analogue of the stronger relation in Eq.~(\ref{bbb}) is not valid as we have seen in Eq.~(\ref{qqq}).

\end{remark}

\subsection{Non-additivity operators as a measure of non-commutativity and quantum interference}
In refs \cite{VV,V1} we have studied the following `non-additivity operator' that quantifies  the amount of quantum intereference, in the sense that it is the difference between the two sides in Eq.~(\ref{qqq}) which as we explained is due to quantum interference:
\begin{eqnarray}\label{567}
{\mathfrak d}(H_1,H_2)=\Pi(H_1\vee H_2)-\Pi(H_1)-\Pi(H_2) .
\end{eqnarray}
 We have shown there that it is related to the commutator of the projectors 
$\Pi(H_1), \Pi(H_2)$ as follows:
\begin{eqnarray}\label{17}
[\Pi(H_1), \Pi(H_2)]={\mathfrak d}(H_1,H_2)[\Pi(H_1)-\Pi(H_2)] ;\;\;\;{\rm Tr}[{\mathfrak d}(H_1,H_2)]=0.
\end{eqnarray}
It is seen that ${\mathfrak d}(H_1,H_2)=0$ (i.e., Eqs.~(\ref{qqq}),(\ref{aaa})  become equalities), if and only if the $\Pi(H_1), \Pi(H_2)$ commute.

The difference between quantum probabilities and classical Kolmogorov probabilities is the non-additivity in Eq.(\ref{aaa}), in contrast to the additivity in Eq.~(\ref{bbb}). It is related to the non-commutativity of the
projectors $\Pi(H_1), \Pi(H_2)$ in a quantum context. Only if the $\Pi(H_1), \Pi(H_2)$ commute, Eq.(\ref{aaa}) becomes equality and the additivity of Kolmogorov probabilities, also holds for quantum probabilities.
In this sense , the ${\mathfrak d}(H_1,H_2)$ is an alternative way to study the quantum non-commutativity.

The ${\mathfrak d}(H_1,H_2)$ is the first step in a ladder of operators defined through M\"obius transforms.
We have studied such operators in a different context in ref.\cite{V2017}, but here we are only interested in  ${\mathfrak d}(H_1,H_2)$.

\section{Coherent projectors}

Let $\ket{s}$ be a `fiducial vector' and $\Pi(0)=\ket{s}\bra{s}$ the corresponding fiducial projector.  
We consider the $d^2$ vectors $D(i)\ket{s}$, which are coherent states in a finite-dimensional Hilbert space \cite{V1}.
The term coherent indicates the properties in proposition \ref{pro12} below.
For a `generic' fiducial vector any $d$ of these states are linearly independent.
We conjecture that any state apart from the position and momentum states, is a generic fiducial vector.

We also consider the corresponding one-dimensional subspaces $H(\{i\})$, and the corresponding `displaced projectors'
\begin{eqnarray}\label{dis}
\Pi(\{i\})=D(i)\Pi(0)[D(i)]^\dagger=D(i)\ket{s}\bra{s}[D(i)]^\dagger;\;\;\;i\in {\mathbb Z}_d\times {\mathbb Z}_d
\end{eqnarray}

Let $A$ be a subset of the phase space ${\mathbb Z}_d\times {\mathbb Z}_d$ with cardinality $|A|$.
We consider the subspace $H(A)$ which is the disjunction of the subspaces $H(\{i\})$ with $i\in A$:
\begin{eqnarray}
H(A)=\bigvee _{i\in A}H(\{i\});\;\;\;A\subseteq{\mathbb Z}_d\times {\mathbb Z}_d
\end{eqnarray}
 $H(A)$ contains vectors of the type
\begin{eqnarray}
\ket{v}=\sum _{i\in A}\lambda _iD(i)\ket{s}.
\end{eqnarray}
$\Pi(A)$ is the corresponding projector, and we  note that
\begin{eqnarray}\label{dfg}
\Pi(A)\ne \sum _{i\in A}\Pi(\{i\}).
\end{eqnarray}
It is easily seen that
\begin{eqnarray}
H(A)\vee H(B)=H(A\cup B).
\end{eqnarray}
Since any $d$ of the  $D(i)\ket{s}$ are linearly independent, it follows that for any set with cardinality $d\le |A|\le d^2$, we get $H(A)=H_d$ and $\Pi(A)={\bf 1}$.

We note that
\begin{eqnarray}
&&{\rm Tr}[\Pi(A)]=|A|\;\;{\rm if}\;\;1\le |A|< d\nonumber\\
&&{\rm Tr}[\Pi(A)]=d\;\;{\rm if}\;\;d\le|A|\le d^2.
\end{eqnarray}
For $A=\emptyset$, we get $H(\emptyset)={\cal O}$ and  $\Pi(\emptyset)=0$.

We refer to $H(A)$ as coherent subspaces, and the $\Pi(A)$ as coherent projectors, because of their properties in proposition \ref{pro12} below.
In ref.\cite{V2016} we have studied coherent subspaces analogous to the above, in the harmonic oscillator context (that involves continuous variables).

In the context of coherent projectors associated with one-dimensional coherent subspaces, the non-additivity operator of Eq.(\ref{567}) is  
\begin{eqnarray}\label{rty}
{\mathfrak d}(\{i,j\})=\Pi(\{i,j\})-\Pi(\{i\})-\Pi(\{j\});\;\;\;i,j \in {\mathbb Z}_d\times {\mathbb Z}_d .
\end{eqnarray}
As we explained earlier, this operator quantifies interference and non-commutativity.

\section{Subsets of the phase space and their displacements}

Let ${\mathfrak N}(a)$ be the set that contains all subsets of the phase space ${\mathbb Z}_d\times {\mathbb Z}_d$ with cardinality $a$:
\begin{eqnarray}
{\mathfrak N}(a)=\{A\;|\;|A|=a\}.
\end{eqnarray}

Also let ${\mathfrak H}(a)$, ${\mathfrak M}(a)$  be the sets that contain the coherent subspaces $H(A)$ and  coherent  projectors $\Pi(A)$ with $A\in {\mathfrak N}(a)$
(of rank $a$), correspondingly:
\begin{eqnarray}
{\mathfrak H}(a)=\{H(A)\;|\;A\in {\mathfrak N}(a)\};\;\;\;{\mathfrak M}(a)=\{\Pi(A)\;|\;A\in {\mathfrak N}(a)\}.
\end{eqnarray}
The ${\mathfrak N}(a)$, ${\mathfrak H}(a)$, ${\mathfrak M}(a)$ have the same cardinality, which is:
\begin{eqnarray}
{\mathfrak n}(a)= 
\begin{pmatrix}
d^2 &a
\end{pmatrix}^T.
\end{eqnarray}
In the special case $a=1$ we get ${\mathfrak n}(1)=d^2$.

\begin{notation}
We consider the subset of the phase space
\begin{eqnarray}
    A=\{(a_1,\beta _1),...,(\alpha_k,\beta _k)\}\subseteq{\mathbb Z}_d\times {\mathbb Z}_d,
\end{eqnarray}
and let $i=(\gamma, \delta)$. We will use the notation
\begin{eqnarray}
A+i=\{(a_1+\gamma,\beta _1+\delta),...,(\alpha_k+\gamma,\beta _k+\delta)\}\subseteq{\mathbb Z}_d\times {\mathbb Z}_d.
\end{eqnarray}
For various $i\in {\mathbb Z}_d\times {\mathbb Z}_d$, the $A$ is displaced in phase space.
$A$ and $A+i$ have the same cardinality. 
\end{notation}
For a given $A$, when $i$ takes all the $d^2$ values of ${\mathbb Z}_d\times {\mathbb Z}_d$, we get  $d^2$ from the
${\mathfrak n}(|A|)$
subsets of the phase space with the same cardinality $|A|$ (but some of them might be equal to each other).
In other words the cardinality of the following set that contains  subsets of the phase space
\begin{eqnarray}
{\mathfrak S}(A)=\{A+i\;|\;i\in{\mathbb Z}_d\times {\mathbb Z}_d\},
\end{eqnarray}
might be less than $d^2$, as the following example shows.
\begin{example}\label{ex23}
We consider a quantum system with variables in  ${\mathbb Z}_3$ and the subsets of the phase space ${\mathbb Z}_3\times {\mathbb Z}_3$ with cardinality $3$. There are ${\mathfrak n}(3)=84$ such subsets. In particular, let
\begin{eqnarray}
  A=\{(0,0),(1,2),(2,1)\}
\end{eqnarray}
We consider the $A+i$ with all $9$ values of $i$, and we get only $3$ sets which are different from each other. 
\begin{eqnarray}\label{61A}
&&A+(0,0)=A+(1,2)=A+(2,1)=\{(0,0),(1,2),(2,1)\}\nonumber\\
&&A+(0,1)=A+(1,0)=A+(2,2)=\{(0,1),(1,0),(2,2)\}\nonumber\\
&&A+(0,2)=A+(1,1)=A+(2,0)=\{(0,2),(1,1),(2,0)\}.
\end{eqnarray}
We note that in the calculation of $A+(0,0)$ we get
\begin{eqnarray}\label{61B}
&&(0,0)+(0,0)=(0,0)\nonumber\\
&&(1,2)+(0,0)=(1,2)\nonumber\\
&&(2,1)+(0,0)=(2,1),
\end{eqnarray}
and in the calculation of $A+(1,2)$ we get
\begin{eqnarray}\label{61C}
&&(0,0)+(1,2)=(1,2)\nonumber\\
&&(1,2)+(1,2)=(2,1)\nonumber\\
&&(2,1)+(1,2)=(0,0).
\end{eqnarray}
Since the  order of the elements in the set $A$ is irrelevant, we get $A+(0,0)=A+(1,2)$.

Therefore $|{\mathfrak S}(A)|=3$.
In this example the cardinality of ${\mathfrak S}(A)$ is smaller than $d^2$.
\end{example}

\begin{lemma}\label{L1}
For odd prime $d$,
 if $|A|\le d-1$ then the cardinality of ${\mathfrak S}(A)$ is $d^2$.
 \end{lemma}
\begin{proof}
We first point out that if $(a_1,a_2)\in {\mathbb Z}_d\times {\mathbb Z}_d$  and $(i_1,i_2)$ takes all $d^2$ values in ${\mathbb Z}_d\times {\mathbb Z}_d$, then $(a_1+i_1,a_2+i_2$ also takes all $d^2$ values in ${\mathbb Z}_d\times {\mathbb Z}_d$.
Therefore if $|A|=1$, then the cardinality of ${\mathfrak S}(A)$ is $d^2$.
We have seen in example \ref{ex23}, that if $|A|\ge 2$ we can have 
\begin{eqnarray}
A+i=A+j;\;\;\;i=(i_1,i_2);\;\;\;j=(j_1,j_2),
\end{eqnarray}
because the order of the elements in the set $A$ is irrelevant.

We next assume that the cardinality of ${\mathfrak S}(A)$ is less than $d^2$, and we will prove that in this case $|A|\ge d$.
If $|{\mathfrak S}(A)|<d^2$,
there must be two different elements $\alpha=(\alpha_1, \alpha _2)$ and $\beta=(\beta _1, \beta _2)$ in $A$,  such that
\begin{eqnarray}
(\alpha_1, \alpha _2)+(i_1,i_2)=(\beta _1, \beta _2)+(j_1,j_2);\;\;\;i=(i_1,i_2);\;\;\;j=(j_1,j_2).
\end{eqnarray}
For example in Eqs.(\ref{61B}), (\ref{61C}),  
\begin{eqnarray}
(\alpha_1, \alpha _2)=(0,0);\;\;\;(i_1,i_2)=(0,0);\;\;\;(\beta _1, \beta _2)=(2,1);\;\;\;(j_1,j_2)=(1,2).
\end{eqnarray}

Therefore
\begin{eqnarray}
(\beta _1, \beta _2)=(\alpha_1+(i_1-j_1), \alpha _2+(i_2-j_2))\in A.
\end{eqnarray}
We apply the same argument to $(\alpha_1+(i_1-j_1), \alpha _2+(i_2-j_2))\in A$ and we conclude that the 
$(\alpha_1+2(i_1-j_1), \alpha_2+2(i_2-j_2))\in A$.
We continue in this way and we find that 
\begin{eqnarray}\label{21}
(\alpha+k (i_1-j_1), \beta+k (i_2-j_2))\in A;\;\;k=0,1,...,d-1.
\end{eqnarray}
In the case of prime $d$ the ${\mathbb Z}_d$ is a field, and consequently Eq.~(\ref{21}) has $d$ pairs different from each other, and the cardinality of $A$ is at least $d$.
In the case of non-prime $d$ the ${\mathbb Z}_d$ is a ring,  it has non-trivial divisors of zero, and  consequently Eq.~(\ref{21}) might have less than $d$ pairs different from each other.

We have assumed that the cardinality of ${\mathfrak S}(A)$ is less than $d^2$, and we proved that the cardinality of $A$ is at least $d$.
Therefore  if $|A|<d$, the cardinality of ${\mathfrak S}(A)$ has the maximum possible value which is $d^2$.
\end{proof}

\section{ Granularisations of the phase space and of the Hilbert space}
\subsection{Equivalence classes: different granular structures}

Below we define three equivalence relations.
They are isomorphic to each other, but we differentiate them with the letters $S, H, P$ which stand for sets, Hilbert spaces,  and projectors correspondingly.
Different equivalence classes define different granular structures of the phase space ${\mathbb Z}_d\times {\mathbb Z}_d$ and of the Hilbert space $H_d$.

\begin{definition}
\mbox{}
\begin{itemize}
\item[(1)]
Let $A,B\in {\mathfrak N}(a)$ (i.e., $|A|=|B|=a$).
$A\overset{S}\sim B$ if there exists $i\in{\mathbb Z}_d\times {\mathbb Z}_d$ such that $B=A+i$,  and therefore both subsets $A,B$ belong to the same set ${\mathfrak S}(A)$.
\item[(2)]
Let $H(A), H(B) \in {\mathfrak H}(a)$  (i.e., $|A|=|B|=a$). Then
$H(A)\overset {H} \sim H(B)$ if $A\overset {S}\sim B$.
\item[(3)]
Let $\Pi(A), \Pi(B) \in {\mathfrak M}(a)$  (i.e., $|A|=|B|=a$). Then
$\Pi(A)\overset {P} \sim \Pi(B)$ if $A\overset {S}\sim B$.
\end{itemize}
\end{definition}

\begin{proposition}\label{pro1}
For odd prime $d$ and $a\le d-1$:
\begin{itemize}
\item[(1)]
$\overset {S} \sim$ is an equivalence relation in the
set ${\mathfrak N}(a)$ of subsets of the phase space, which  partitions it
 into ${\mathfrak n}(a)/d^2$ equivalence classes denoted as  ${\mathfrak N}(\nu;a)$, with the equivalence class label $\nu$ taking the values $\nu=1,...,{\mathfrak n}(a)/d^2$.
Each equivalence class has $d^2$ elements.
 \item[(2)]
$\overset {H} \sim$ is an equivalence relation in the
set ${\mathfrak H}(a)$ of coherent subspaces, which  partitions it
 into ${\mathfrak n}(a)/d^2$ equivalence classes denoted as  ${\mathfrak H}(\nu;a)$, with with the equivalence class label $\nu$ taking the values $\nu=1,...,{\mathfrak n}(a)/d^2$.
Each equivalence class has $d^2$ elements.
 \item[(3)]
$\overset {P} \sim$ is an equivalence relation in the
set ${\mathfrak M}(a)$ of coherent projectors, which  partitions it
 into ${\mathfrak n}(a)/d^2$ equivalence classes denoted as  ${\mathfrak M}(\nu;a)$, with with the equivalence class label $\nu$ taking the values $\nu=1,...,{\mathfrak n}(a)/d^2$.
Each equivalence class has $d^2$ elements.

\end{itemize}
\end{proposition}
\begin{proof}
\mbox{}
\begin{itemize}
\item[(1)]
We first point out that  ${\mathfrak n}(a)/d^2$ is an integer.
Indeed the general binomial coefficient  
$\begin{pmatrix}
m\\
k\\
\end{pmatrix}$ is divisible by $\frac{m}{{\rm GCD}(m,k)}$. 
In our case $m=d^2$ and $k=a$.
For prime $d$ and $a<d$, the ${\rm GCD}(d^2,a)=1$. Therefore the  ${\mathfrak n}(a)/d^2$ is an integer.

We next prove that $\overset {S}\sim$ is an equivalence relation. The following properties hold:
\begin{eqnarray}
&&A\overset {S}\sim A\nonumber\\
&&A\overset {S}\sim B\;\rightarrow\;B\overset {S}\sim A\nonumber\\
&&A\overset {S}\sim B\;{\rm and}\;B\overset {S}\sim C\;\rightarrow\;A\overset {S}\sim C.
\end{eqnarray}
Therefore $\overset {S}\sim$ is an equivalence relation that partitions the set ${\mathfrak N} (a)$.
The cardinality of ${\mathfrak N} (a)$ is ${\mathfrak n} (a)$, and we proved in lemma \ref{L1} that  in the case $a<d-1$ the cardinality of each equivalence class is $d^2$.
Therefore the number of equivalence classes is ${\mathfrak n}(a)/d^2$.
\item[(2)]
Analogous arguments to the first part, also hold for the equivalence relation $\overset {H}\sim$.

\item[(3)]
Analogous arguments to the first part, also hold for the equivalence relation $\overset {P}\sim$.

\end{itemize}
\end{proof}

\begin{remark}
In the special case $a=1$, the  ${\mathfrak N}(1)$, ${\mathfrak H}(1)$, ${\mathfrak M}(1)$ consist of one equivalence class.
Also for $|A|\ge d$ we get $\Pi(A)={\bf 1}$, and for this reason below we study the case $ a \le d-1$.
\end{remark}

\subsection{Coverings of the phase space and coverings of the Hilbert space}
\begin{proposition}\label{pro2}
Let $d$ be an odd prime and $a\le d-1$.
The  $d^2$ subsets within a given equivalence class ${\mathfrak N}(\nu;a)$, cover the phase space ${\mathbb Z}_d\times {\mathbb Z}_d$ in the sense that that their union is
 \begin{eqnarray}\label{4R}
\bigcup _{A\in {\mathfrak N}(\nu;a)}A={\mathbb Z}_d\times {\mathbb Z}_d.
\end{eqnarray}

Also the $d^2$ coherent subspaces  within a given equivalence class ${\mathfrak H}(\nu;a)$ 
(and the corresponding projectors in the equivalence class ${\mathfrak M}(\nu;a)$) cover the Hilbert space $H_d$, in the sense that their disjunction is
 \begin{eqnarray}\label{4RR}
\bigvee _{H(A)\in {\mathfrak H}(\nu;a)}H(A)=H_d.
\end{eqnarray}
\end{proposition}
\begin{proof}
We consider any element $(a_1,a_2)\in A$.  We have explained earlier that as $(i_1,i_2)$ takes all $d^2$ values in ${\mathbb Z}_d\times {\mathbb Z}_d$, then $(a_1+i_1,a_2+i_2)$ also takes all $d^2$ values in ${\mathbb Z}_d\times {\mathbb Z}_d$.
This proves Eq.(\ref{4R}).
Also Eq.(\ref{4R}) and the fact that all coherent states span the whole Hilbert space $H_d$ proves Eq.(\ref{4RR}). 
\end{proof}
The $d^2$ coherent projectors in a given equivalence class  ${\mathfrak M}(\nu;a)$ cover the Hilbert space not only in the sense of  Eq.(\ref{4RR}) for the corresponding coherent subspaces, but more importantly (because it is stronger) in the sense 
of the resolution of the identity which will be given later in  Eq.~(\ref{2A}), and which imples that every state $\ket{s}$ in $H_d$ can be expanded as
\begin{eqnarray}\label{2A}
\ket{s}=\frac{1}{d|A|}\sum _i\Pi(A+i)\ket {s},
\end{eqnarray}
where $\Pi(A+i)\ket {s}$ are vectors in  ${\mathfrak H}(\nu;a)$.

\begin{remark}
We have introduced a granular structure for the phase space ${\mathbb Z}_d\times {\mathbb Z}_d$, with the cardinality of the sets $A$ defining the size of the grains.
Mathematically this  is not a partition of the phase space, but a covering of it.
Because for sets in an
equivalence classes ${\mathfrak N}(\nu;a)$ we might have non-zero overlaps ($A\cap B\ne \emptyset$).
Different equivalence classes correspond to different granular structures and lead to different coverings of the phase space.
Starting from a set within a given equivalence class and displacing it, we get the $d^2$ sets within that equivalence class.  

We have also introduced a  granular structure for the Hilbert space $H_d$, and in quantum measurements with the projectors $\Pi(A)$ the ${\rm Tr}[\Pi(A)]$ defines the size of the grains.
But again this  is not a partition of the Hilbert space, but a covering of it.
Because for subspaces in an
equivalence classes ${\mathfrak H}(\nu;a)$ we might have non-zero overlaps ($H(A)\wedge H(B)\ne{\cal O}$).
Different equivalence classes correspond to different granular structures and lead to different coverings of the Hilbert space.
Starting from a projector within a given equivalence class and displacing it as in Eq.(\ref{dis}), we get the $d^2$ projectors within that equivalence class.  
\end{remark}

\subsection{Partial ordering of the coverings: fine and coarse coverings}

\begin{definition}\label{def10}
The covering of the phase space ${\mathbb Z}_d\times {\mathbb Z}_d$ by the sets in the equivalence class ${\mathfrak N}(\nu;a)$ is finer than the covering  by the sets in the equivalence class ${\mathfrak N}(\mu;b)$, if for some set $A\in {\mathfrak N}(\nu;a)$
there exist a set $B\in {\mathfrak N}(\mu;b)$ such that $A\subseteq B$. We denote this as ${\mathfrak N}(\nu;a) \overset {S}\sqsubset {\mathfrak N}(\mu;b)$ (the superfix $S$ stands for sets).
In this case we also say that the covering by the sets in the equivalence class ${\mathfrak N}(\mu;b)$ is coarser than the covering  by the sets in the equivalence class ${\mathfrak N}(\nu;a)$.
\end{definition}
\begin{proposition}\label{pro3}
\mbox{}
\begin{itemize}
\item[(1)]
If ${\mathfrak N}(\nu;a) \overset {S} \sqsubset {\mathfrak N}(\mu;b)$ then there is a bijective map between the $d^2$ sets $A_i\in {\mathfrak N}(\nu;a)$ and the $d^2$ sets $B_i\in {\mathfrak N}(\mu;b)$ such that $A_i\subseteq B_i$.
\item[(2)]
$ \overset {S}\sqsubset$ is a partial order among the equivalence classes of subsets of the phase space.
\item[(3)]
${\mathfrak N}(1;1) \overset {S}\sqsubset {\mathfrak N}(\nu;a)\overset {S}\sqsubset {\mathfrak N}(\nu;d-1)$ and therefore
the finest and (non-trivial) coarsest  coverings of the phase space are by the sets in the equivalence classes  ${\mathfrak N}( 1;1)$ and ${\mathfrak N}(\nu;d-1)$, correspondingly.
\end{itemize}
\end{proposition}
\begin{proof}
\mbox{}
\begin{itemize}
\item[(1)]
If ${\mathfrak N}(\nu;a)  \overset {S}\sqsubset {\mathfrak N}(\mu;b)$ then for some set $A\in {\mathfrak N}(\nu;a)$
there exist a set $B\in {\mathfrak N}(\mu;b)$ such that $A\subseteq B$.
We label the $d^2$ elements of ${\mathfrak N}(\nu;a)$ as $A_i=A+i$ and the  $d^2$ elements of ${\mathfrak N}(\mu;b)$ as $B_i=B+i$, where $i\in {\mathbb Z}_d\times {\mathbb Z}_d$.
The fact that $A\subseteq B$ implies that $A+i\subseteq B+i$.
We note that since $d$ is an odd prime and $a<d$, the map between $A_i$ and $ B_i$ is bijective. This proves the statement. 
\item[(2)]
It is easily seen that
\begin{eqnarray}
&&{\mathfrak N}(\nu;a) \overset {S} \sqsubset {\mathfrak N}(\nu;a)\nonumber\\
&&{\mathfrak N}(\nu;a)  \overset {S}\sqsubset {\mathfrak N}(\mu;b)\;\;{\rm and}\;\;{\mathfrak N}(\mu;b)  \overset {S}\sqsubset {\mathfrak N}(\nu;a)\;\;\rightarrow\;\;{\mathfrak N}(\nu;a)= {\mathfrak N}(\mu;b)\nonumber\\
&&{\mathfrak N}(\nu;a)  \overset {S}\sqsubset {\mathfrak N}(\mu;b)\;\;{\rm and}\;\;{\mathfrak N}(\mu;b)  \overset {S}\sqsubset {\mathfrak N}(\kappa;c)\;\;\rightarrow\;\;{\mathfrak N}(\nu;a) \overset {S}\sqsubset {\mathfrak N}(\kappa;c)
\end{eqnarray}
This proves that $ \overset {S}\sqsubset$ is a partial order among the equivalence classes  of subsets of the phase space.
\item[(3)]
It is easily seen that ${\mathfrak N}(1,1) \overset {S}\sqsubset {\mathfrak N}(\nu,a)$.
\end{itemize}
\end{proof}
\begin{definition}\label{def11}
\mbox{}
\begin{itemize}
\item[(1)]
The covering of the Hilbert space $H_d$ by the coherent subspaces in the equivalence class ${\mathfrak H}(\nu;a)$ is finer than the covering  by the coherent subspaces in the equivalence class ${\mathfrak H}(\mu;b)$, if 
${\mathfrak N}(\nu;a) \overset {S}\sqsubset {\mathfrak N}(\mu;b)$ for the corresponding equivalence classes of subsets of the phase space.
In this case there is a bijective map between the $d^2$ subspaces $H(A_i)\in {\mathfrak H}(\nu;a)$ and the $d^2$ subspaces $H(B_i)\in {\mathfrak H}(\mu;b)$ such that  $H(A_i)\overset {H}\prec H(B_i)$.
We denote this as ${\mathfrak H}(\nu;a) \overset {H}\sqsubset {\mathfrak H}(\mu;b)$.
\item[(2)]
The covering of the Hilbert space $H_d$ by the coherent  projectors in the equivalence class ${\mathfrak M}(\nu;a)$ is finer than the covering  by the coherent projectors in the equivalence class ${\mathfrak M}(\mu;b)$, if 
${\mathfrak N}(\nu;a) \overset {S}\sqsubset {\mathfrak N}(\mu;b)$ for the corresponding equivalence classes of subsets of the phase space.
In this case there is a bijective map between the $d^2$ projectors $\Pi(A_i)\in {\mathfrak M}(\nu;a)$ and the $d^2$ projectors $\Pi(B_i)\in {\mathfrak M}(\mu;b)$ such that  $\Pi(A_i)\overset {P}\prec \Pi(B_i)$.
We denote this as ${\mathfrak M}(\nu;a) \overset {P}\sqsubset {\mathfrak M}(\mu;b)$.
\end{itemize}
\end{definition}
Clearly $ \overset {H}\sqsubset $ and  $ \overset {P}\sqsubset $ are partial orders  isomorphic to $ \overset {S}\sqsubset $.
The finest covering of the Hilbert space is by the coherent subspaces in the equivalence class  ${\mathfrak H}( 1;1)$ (and by the corresponding coherent projectors in the equivalence class  ${\mathfrak M}( 1;1)$).
The (non-trivial) coarsest covering of the Hilbert space is by the coherent subspaces  in the equivalence class  ${\mathfrak H}( \nu;d-1)$ (and by the corresponding coherent projectors in the equivalence class  ${\mathfrak M}( \nu;d-1)$).

Finer coverings are closer to classical physics in the sense that they are close to the classical phase space.
Coarser coverings use larger subspaces which involve the concept of superposition, and in this sense they are closer to quantum physics. 

\section{Properties of the coherent projectors}

\begin{proposition}\label{pro12}
For odd prime $d$ and $a\le d-1$:
\begin{itemize}
\item[(1)]

The subspaces $H(A)$, $H(A+i)$ within a given equivalence class  ${\mathfrak H}(\nu;a)$ are isomorphic and they are related through the following map that involves displacement transformations:
\begin{eqnarray}
D(i)\ket{s}=\ket{t};\;\;\;\ket{s}\in H(A);\;\;\;\ket{t}\in H(A+i).
\end{eqnarray}
We denote this as $D(i)H(A)=H(A+i)$.

The subspaces in different equivalence classes within ${\mathfrak H}(a)$ are not related through  displacement transformations (although  they are isomorphic to each other because they have the same dimension) .
\item[(2)]
Within a given equivalence class  ${\mathfrak M}(\nu;a)$, the following closure property holds for coherent projectors  that involves displacement transformations:
\begin{eqnarray}\label{33}
D(i)\Pi(A)[D(i)]^\dagger=\Pi(A+i);\;\;\;i=(i_1,i_2)\in {\mathbb Z}_d\times {\mathbb Z}_d.
\end{eqnarray}
Projectors in different equivalence classes are not related with displacement transformations.
\item[(3)]
The $d^2$ coherent projectors in a given equivalence class  ${\mathfrak M}(\nu;a)$, resolve the identity as follows:
\begin{eqnarray}\label{2A}
\frac{1}{d|A|}\sum _i\Pi(A+i)={\bf 1};\;\;\;i=(i_1,i_2)\in {\mathbb Z}_d\times {\mathbb Z}_d.
\end{eqnarray}

In the special case $a=1$ we have only one equivalent class, and we get
\begin{eqnarray}
\frac{1}{d}\sum _i\Pi(i)={\bf 1};\;\;\;i=(i_1,i_2)\in {\mathbb Z}_d\times {\mathbb Z}_d.
\end{eqnarray}
\item[(4)]

 The following resolutions of the identity involves all  ${\mathfrak n}(a)$ coherent projectors in ${\mathfrak M}(a)$ (in all equivalent classes ${\mathfrak M}(\nu;a)$ with $\nu=1,...,{\mathfrak n}(a)/d^2$):
 \begin{eqnarray}\label{2AA}
\frac{d}{a{\mathfrak n}(a)}\sum _{\Pi(A)\in {\mathfrak M}(a)}\Pi(A)={\bf 1}
\end{eqnarray}
\end{itemize}
\end{proposition}
\begin{proof}
\mbox{}
\begin{itemize}
\item[(1)]
By definition $H(A)$ is the subspace spanned by the vectors $D(j)\ket{s}$ with $j\in A$.
Acting on these vectors with $D(i)$ we get the vectors $\lambda_{ij}D(i+j)\ket{s}$ (where $\lambda _{ij}$ is a phase factor) which by definition span the space $H(A+i)$.

\item[(2)]
The fact that $D(i)H(A)=H(A+i)$ leads to  Eq.~(\ref{33}).

For projectors $\Pi(A), \Pi(B)$ in different equivalence classes, there is no $i\in  {\mathbb Z}_d\times {\mathbb Z}_d$ such that $A+i=B$.

Using the first of Eqs.(\ref{33}), we easily prove the second one. 
\item[(3)]
We use Eq.(\ref{1A}) with $\theta=\Pi(A)$, ${\rm Tr} \theta=|A|$, and taking into account Eq.(\ref{33}) we prove the first of Eqs.(\ref{2A}).
The second equation follows immediately from the first one.

\item[(4)]

We add Eqs.(\ref{2A}) for all $\nu$, and we get  Eq.~(\ref{2AA}).

\end{itemize}

\end{proof}

\begin{proposition}
\mbox{}
\begin{itemize}
\item[(1)]
Within a given equivalence class  ${\mathfrak M}(\nu;a)$, the 
\begin{eqnarray}\label{41}
{\rm Tr}[\Pi(A+i)\Pi(A+j)]=\xi_{\nu,a}(i-j);\;\;\;i,j\in {\mathbb Z}_d\times {\mathbb Z}_d.
\end{eqnarray}
depends only on the difference $i-j$ (it does not depend on $i+j$).
\item[(2)]
\begin{eqnarray}\label{kk}
\sum _j\xi_{\nu,a}(i-j)=da^2;\;\;\;j\in {\mathbb Z}_d\times {\mathbb Z}_d.
\end{eqnarray}
In the special case $a=1$ we get
\begin{eqnarray}
\sum _j\xi_{1,1}(i-j)=d;\;\;\;j\in {\mathbb Z}_d\times {\mathbb Z}_d.
\end{eqnarray}
\end{itemize}

\end{proposition}
\begin{proof}
\mbox{}
\begin{itemize}
\item[(1)]
It is easily seen that
\begin{eqnarray}
&&{\rm Tr}[\Pi(A+i)\Pi(A+j)]={\rm Tr}[D(i)\Pi(A)D(-i)D(j)\Pi(A)D(-j)]\nonumber\\&&={\rm Tr}[D(i-j)\Pi(A)D(-i+j)\Pi(A)]={\rm Tr}[\Pi(A+i-j)\Pi(A)]
\end{eqnarray}
Therefore this trace depends only on $i-j$.
\item[(2)]
Eq.(\ref{kk}) is proved by taking the summation over $j$ in Eq.(\ref{41}) and using Eq.(\ref{2A}).
\end{itemize}
\end{proof}

The following proposition is for the case that $|A|=2$.
It shows that the average of the non-additivity operator ${\mathfrak d}_{\nu}(\{i_1,i_2\})$ within an equivalence class, is zero.
This can be interpreted as the fact that constructive interference is equal to the destructive interference and the average result within an equivalence class, is zero.

\begin{proposition}
Let $d$ be an  odd prime and $a\le d-1$.
We denote as  ${\mathfrak d}_{\nu}(\{i,j\})$ the $d^2$ non-additivity operators ${\mathfrak d}(\{i,j\})$ in Eq.(\ref{rty}) with $\{i,j\}$ in the equivalence class  ${\mathfrak N}(\nu;2)$.
Then :
\begin{eqnarray}\label{bn}
\sum {\mathfrak d}_{\nu}(\{i,j\})=0;\;\;\;\{i,j\}\in {\mathfrak N}(\nu;2).
\end{eqnarray}
The summation involves all $d^2$ sets in the equivalence class ${\mathfrak N}(\nu;2)$.

\end{proposition}
\begin{proof}

Within a given equivalence class ${\mathfrak M}(\nu;2)$
\begin{eqnarray}
\sum _k\Pi(\{i,j\}+k)=2d{\bf 1}.
\end{eqnarray}
Also 
\begin{eqnarray}
\sum _k\Pi(i+k)=\sum _k\Pi(j+k)=d{\bf 1}.
\end{eqnarray}
Therefore
\begin{eqnarray}
\sum _k[\Pi(\{i,j\}+k)-\Pi(i+k)-\Pi(j+k)]=0.
\end{eqnarray}
This proves Eq.(\ref{bn}).

\end{proof}

\section{$Q$ and $P$ functions  based on coherent projectors in a given equivalence class} 
Given a Hermitian operator $\theta$ and an equivalence class ${\mathfrak N}(\nu, a)$, we define the generalized $Q$-function as the following set of $d^2$  numbers:
 \begin{eqnarray}\label{Q}
{\mathfrak Q}[\theta, {\mathfrak N}(\nu;a)]=\{Q_{\nu,a}(A;\theta)\;|\;A\in {\mathfrak N}(\nu; a)\};\;\;\;Q_{\nu,a}(A;\theta)={\rm Tr}[\theta \Pi(A)].
\end{eqnarray}
In particular for a density matrix $\rho$
 \begin{eqnarray}\label{Q}
Q_{\nu,a}(A;\rho)={\rm Tr}[\rho \Pi(A)];\;\;\;A\in {\mathfrak N}(\nu; a).
\end{eqnarray}
$Q_{\nu,a}(A;\rho)$ are the probabilities that measurement with the projector $\Pi(A)$ on the system described by $\rho$,  will give the outcome `yes'.
We note that the set ${\mathfrak Q}[\rho, {\mathfrak N}(\nu; a)]$ with appropriate normalization is a quasi-probability distribution, because the $\Pi(A)$ do not commute with each other.

The partial order $\overset {P}\sqsubset$ endows a partial order  $\overset {Q}\sqsubset$ on the sets ${\mathfrak Q}[\rho, {\mathfrak N}(\nu; a)]$. 
\begin{definition}
The $Q$ function ${\mathfrak Q}[\rho, {\mathfrak N}(\nu; a)]$ is finer than ${\mathfrak Q}[\rho, {\mathfrak N}(\mu; b)]$ 
if ${\mathfrak N}(\nu;a) \overset {S}\sqsubset {\mathfrak N}(\mu;b)$.
 In this case there is a bijective map between $Q_{\nu,a}(A_i;\rho)\in {\mathfrak Q}[\rho, {\mathfrak N}(\nu;a)]$ and $Q_{\mu,b}(B_i;\rho)\in {\mathfrak Q}[\rho, {\mathfrak N}(\mu;b)]$  such that $Q_{\nu,a}(A_i;\theta)\le Q_{\mu,b}(B_i;\theta)$.
We denote this as 
 ${\mathfrak Q}[\rho, {\mathfrak N}(\nu; a)]\overset {Q}\sqsubset {\mathfrak Q}[\rho, {\mathfrak N}(\mu; b)]$.
\end{definition}
A Hermitian operator $\theta$ can be written in terms of the $d^2$ projectors in an equivalence class ${\mathfrak M}(\nu;a)$ as
\begin{eqnarray}\label{ff}
\theta=\sum _iP_{\nu,a}(A+i;\theta)\Pi(A+i);\;\;\;i=(i_1,i_2)\in {\mathbb Z}_d\times {\mathbb Z}_d.
\end{eqnarray}
Indeed from this follows the following system of $d^2$ equations with $d^2$ unknowns:
\begin{eqnarray}
\bra{X;k}\theta\ket{X;\ell}=\sum _iP_{\nu,a}(A+i;\theta)\bra{X;k}\Pi(A+i)\ket{X;\ell};\;\;\;i=(i_1,i_2)\in {\mathbb Z}_d\times {\mathbb Z}_d;\;\;\;k,\ell \in {\mathbb Z}_d.
\end{eqnarray}
We solve this and we get the coefficients $P_{\nu,a}(A+i;\theta)$ which are the $P$-representation of the Hermitian operator $\theta$ in the present context. 

\begin{proposition}\label{pro135}
Let  ${\mathfrak M}(\nu;a)$ be an equivalence class with $a\le d-1$, $\theta$ a Hermitian operator and $\rho$ a density matrix. The following relations that involve  summations over all $d^2$ sets in the equivalence class ${\mathfrak N}(\nu;a)$, hold:
\begin{itemize}
\item[(1)]
\begin{eqnarray}\label{62}
{\rm Tr}(\theta)=\frac{1}{da}\sum _iQ_{\nu,a}(A+i;\theta)=a\sum _iP_{\nu,a}(A+i;\theta);\;\;\;i=(i_1,i_2)\in {\mathbb Z}_d\times {\mathbb Z}_d.
\end{eqnarray}
\item[(2)]
\begin{eqnarray}\label{zzz}
{\rm Tr}(\rho\theta)=\sum _iP_{\nu,a}(A+i;\theta)Q_{\nu,a}(A+i;\rho);\;\;\;i=(i_1,i_2)\in {\mathbb Z}_d\times {\mathbb Z}_d.
\end{eqnarray}
\item[(3)]
\begin{eqnarray}\label{zz}
Q_{\nu,a}(A;\rho)=\sum _iP_{\nu,a}(A+i;\theta)\xi_{\nu,a}(i);\;\;\;i=(i_1,i_2)\in {\mathbb Z}_d\times {\mathbb Z}_d.
\end{eqnarray}
The $\xi_{\nu,a}(i)$ has been defined in Eq.(\ref{41}).
\end{itemize}
\end{proposition}
\begin{proof}
\mbox{}
\begin{itemize}
\item[(1)]
The first part is proved using Eq.~(\ref{2A}). 
The second part follows from Eq.~(\ref{ff}) taking into account that ${\rm Tr}[\Pi(A+i)]=|A|$ (for $|A|<d$).
\item[(2)]
We multiply both sides of Eq.~(\ref{ff}) by $\rho$ and we take the trace. This gives Eq.~(\ref{zzz}).
\item[(3)]
We multiply both sides of Eq.~(\ref{ff}) (with $\theta=\rho$) by  $\Pi(A)$ and we take the trace. Using Eq.~(\ref{41}), we prove Eq.~(\ref{zz}).
\end{itemize}
\end{proof}

\subsection{Non-additivity of the $Q$-functions and quantum interference}
Earlier we linked Eq.~(\ref{qqq}) (and also Eq.~(\ref{dfg})) with quantum interference. In terms of the $Q$-functions this is the relation
\begin{eqnarray}
Q_{\nu,a}(A;\theta)\ne \sum _{i\in A}Q_{1,1}(\{i\};\theta).
\end{eqnarray}
Due to quantum interference the information in $Q_{\nu,a}(A;\theta)$ is different from the information in the set $\{Q_{1,1}(\{i\};\theta)|i\in A\}$.

For simplicity we study in more detail the case $|A|=2$. For $\{i,j\}$ in the equivalence class ${\mathfrak N}(\nu;2)$, we define the 
 \begin{eqnarray}\label{456}
\delta Q_{\nu,a}(\{i,j\};\theta)=Q_{\nu,a}(\{i,j\};\theta)-Q_{1,1}(\{i\};\theta)-Q_{1,1}(\{j\};\theta)={\rm Tr}[\theta {\mathfrak d}_{\nu}(\{i,j\})].
\end{eqnarray}
It involves probabilities which are not simultaneously measurable, 
and which need to be measured using different ensembles describing the same density matrix $\rho$.

Using Eq.(\ref{bn}) we prove that
\begin{eqnarray}\label{67}
\sum \delta Q_{\nu,a}(\{i,j\};\theta)=0;\;\;\;\{i,j\}\in {\mathfrak N}(\nu;2).
\end{eqnarray}
The summation involves all $d^2$ sets in the equivalence class ${\mathfrak N}(\nu;2)$.

The $\delta Q_{\nu}(\{i,j\};\rho)$ is intimatelly related to quantum interference and the non-commutativity of $\Pi(\{i\}), \Pi(\{j\})$.
 Within an equivalence class, constructive and destructive interference cancel each other, and the average value of $\delta Q_{\nu}(\{i,j\};\rho)$ is zero.

We note that interference in phase space has been studied from  different points of view in \cite{DSW,A}.

\subsection{Entropies for $Q$-functions}
Let $\rho$ be a density matrix.
Given an equivalence class ${\mathfrak N}(\nu;a)$ of subsets of the phase space, we consider the  $d^2$ non-negative numbers 
\begin{eqnarray}
p_i=\frac{1}{da}Q_{\nu,a}(A+i;\rho);\;\;\;\sum _ip_i=1.
\end{eqnarray}
They form a quasi-probability distribution with entropy 
\begin{eqnarray}\label{ent}
E_{\nu, a}=-\sum _ip_i\log p_i;\;\;\;i\in {\mathbb Z}_d\times {\mathbb Z}_d.
\end{eqnarray}
We use natural logarithms and then the result is in nats.
The entropy quantifies how close to uniform is the distribution $\{p_i\}$.
It takes the maximum value $E_{\nu, a}=\log(d^2)$ when $p_i=\frac{1}{d^2}$.

\section{Lorenz values for the $Q$-functions}

Lorenz values and the Gini index are used extensively in Mathematical Economics for the study of inequality in the distribution of wealth.
We propose similar quantities for the study of inequality in the distribution of the  $Q$-function in the various parts of the phase space.
 They show how spread is the quantum state in the granular structure of the phase space that consists of coherent subspaces in ${\mathfrak H}(\nu, a)$.

Lorenz values and the Gini index are based on  ordering the $Q$ function, and this leads to the concept of the ranking permutation of a density matrix which is a type of postcode that describes its location in the Hilbert space.
We have used this concept previously in refs \cite{C1,C2} in connection with Choquet integrals.
We have also introduced comonotonic density matrices that have the same ranking permutation, again in connection with Choquet integrals.
In the present paper we work with coherent projectors and use these concepts in connection with the Lorenz values and the Gini index for $Q$ functions.

\subsection{The ranking permutation of a density matrix and comonotonicity}

We consider the set ${\mathfrak Q}[\rho, {\mathfrak N}(\nu;a)]$ of the $Q$-functions for a density matrix $\rho$ related to the equivalence class ${\mathfrak N}(\nu;a)$.
We order them from the lowest to the highest as follows:
\begin{eqnarray}\label{order}
Q_{\nu,a}(A_1;\rho)\le Q_{\nu,a}(A_2;\rho)\le ...\le Q_{\nu,a}(A_{d^2};\rho);\;\;\;A_i\in {\mathfrak N}(\nu;a)
\end{eqnarray}
Here we label the subsets in the  equivalence class ${\mathfrak N}(\nu;a)$ according to the value of the corresponding $Q$-function: the index $1$ is given to the subset with the lowest $Q$-function
 (`poorest' subset), the index $2$ to the subset with the second lowest $Q$-function (`second poorest' subset', etc, until the 
 index $d^2$ which is given to the subset with the highest $Q$-function (`richest' subset).
  
Starting from a reference (fixed) ordering $(B_1,...,B_{d^2})$ of the $d^2$ subsets of phase space in the equivalence class ${\mathfrak N}(\nu;a)$, we get any other ordering with permutations:
\begin{eqnarray}
(B_1,...,B_{d^2})\;\overset{\pi} \rightarrow\;(\pi(B_1),...,\pi(B_{d^2})).
\end{eqnarray}
$\pi$ is an element of the symmetric group ${\cal S}_{d^2}$ \cite{sagan} which has $d^2!$ elements.
Multiplication in this group is the composition.
\begin{definition}
The ranking permutation of a density matrix $\rho$ with respect to the covering of the Hilbert space defined by ${\mathfrak H}(\nu;a)$,  is the permutation $\pi$ such that 
\begin{eqnarray}\label{order}
Q_{\nu,a}(\pi(B_1);\rho)\le Q_{\nu,a}(\pi(B_2);\rho)\le ...\le Q_{\nu,a}(\pi(B_{d^2});\rho);\;\;\;\pi(B_i)\in {\mathfrak N}(\nu;a).
\end{eqnarray}
\end{definition}

The ranking permutation indicates the coherent subspaces where the quantum state `mostly lives' and in this sense it is is a kind of postcode of the density matrix within the Hilbert space. 
A quantum particle `lives' everywhere in the Hilbert space, but in some parts more than in others. We quantify this with
the ranking permutation, which is based on ordering the probabilities ${\rm Tr}[\rho \Pi(A_i)]$. 
Comonotonic density matrices `live' in the same part of the Hilbert space and in this sense they are physically similar.

Let ${R}$ be the set of all density matrices, and ${\mathfrak R}$ the subset of ${R}$ that contains the density matrices for which the inequalities in Eq.(\ref{order}) are strict inequalities (i.e., their values of the $Q$ function are different from each other). 
The ranking permutation is unique for density matrices in  ${\mathfrak R}$, and multi-valued for density matrices in  $R\setminus {\mathfrak R}$ (which can be called `borderline' density matrices).

\begin{definition}
Two density matrices in ${\mathfrak R}$ are called comonotonic if they have  the same ranking permutation.
We denote this as $\rho  \overset{\mathfrak R}\sim \sigma$.
\end{definition}
\begin{lemma}
Comonotonicity  is an equivalence relation in ${\mathfrak R}$ (but not in ${ R}$).
\end{lemma}
\begin{proof}
It is easily seen that 
\begin{eqnarray}
&&\rho\overset {\mathfrak R}\sim \rho\nonumber\\
&&\rho\overset {\mathfrak R}\sim \sigma\;\rightarrow\;\sigma \overset {\mathfrak R}\sim \rho\nonumber\\
&&\rho\overset {\mathfrak R}\sim \sigma \;{\rm and}\;\sigma\overset {\mathfrak R}\sim \tau\;\rightarrow\;\rho\overset {\mathfrak R}\sim \tau.
\end{eqnarray}
We note that transitivity is valid in ${\mathfrak R}$ (but not in ${R}$).
\end{proof}
The equivalence relation $\overset{\mathfrak R}\sim$ partitions  ${\mathfrak R}$ into $d^2!$ equivalence classes of comonotonic density matrices, which we label with the  permutation of the ranking permuation of this 
equivalence class,  as ${\mathfrak R}(\nu,a,\pi)$.

\begin{proposition}\label{pro34}
If $\rho, \sigma$ are comonotonic density matrices (i.e., they belong to the same equivalence class ${\mathfrak R}(\nu,a,\pi)$)
then the density matrix $p \rho+(1-p)\sigma$ where $0\le p \le 1$ is comonotonic to both of them (i.e., it also belongs to the equivalence class ${\mathfrak R}(\nu,a,\pi)$).
Therefore ${\mathfrak R}(\nu,a,\pi)$ is a convex set of density matrices. 
\end{proposition}
\begin{proof}
We first point out that
\begin{eqnarray}
Q[p \rho+(1-p)\sigma]=pQ(\rho)+(1-p)Q(\sigma).
\end{eqnarray}
Since $\rho, \sigma$ belong to the same equivalence class, we have 
\begin{eqnarray}
&&Q_{\nu,a}(A_1;\rho)< Q_{\nu,a}(A_2;\rho)< ...< Q_{\nu,a}(A_{d^2};\rho);\;\;\;A_i\in {\mathfrak N}(\nu;a)\nonumber\\
&&Q_{\nu,a}(A_1;\sigma)< Q_{\nu,a}(A_2;\sigma)< ...< Q_{\nu,a}(A_{d^2};\sigma).
\end{eqnarray}
From this follows that
\begin{eqnarray}
pQ_{\nu,a}(A_1;\rho)+(1-p)Q_{\nu,a}(A_1;\sigma)<p Q_{\nu,a}(A_2;\rho)+(1-p) Q_{\nu,a}(A_2;\sigma)< ...
\end{eqnarray}
which is rewriten as
\begin{eqnarray}
&&Q_{\nu,a}[A_1;p\rho+(1-p)\sigma]< Q_{\nu,a}[A_2;p\rho+(1-p)\sigma]< ...< Q_{\nu,a}[A_{d^2};p\rho+(1-p)\sigma].
\end{eqnarray}
This proves that the ranking permutation of $p \rho+(1-p)\sigma$ is the same as  the ranking permutation of $ \rho, \sigma$.
From this follows that  ${\mathfrak R}(\nu,a,\pi)$ is a convex set of density matrices. 

\end{proof}

Every  unitary transformation $U\rho U^{\dagger}$ of the density matrix $\rho$ changes its ranking permuation.
For example if $\pi$ is the ranking permuation of $\rho$, the displaced density matrix $\rho ^{\prime}=D(i)\rho[ D(i)]^{\dagger}$
has ranking permutation the composition $\varpi\circ \pi$ where $\varpi(A)=A+i$ with $i\in {\mathbb Z}_d\times {\mathbb Z}_d$.
Indeed,
 \begin{eqnarray}\label{Q12}
Q_{\nu,a}(A^{\prime};\rho^{\prime})={\rm Tr}[D(i)\rho [D(i)]^\dagger\Pi(A^{\prime})]={\rm Tr}[\rho\Pi(A^{\prime}-i)]=Q_{\nu,a}(A;\rho);\;\;\;A^{\prime}-i=A.
\end{eqnarray}

\subsection{Lorenz values for $Q$ functions}

Given a density matrix $\rho$, we order the values of its $Q$ function as in Eq.(\ref{order}) and then define its Lorenz values ${\cal L}_{\nu,a}(\ell;\rho)$ as follows:
\begin{eqnarray}\label{71}
{\cal L}_{\nu,a}(\ell;\rho)=\frac{1}{da}[Q_{\nu,a}(A_1;\rho)+Q_{\nu,a}(A_2;\rho)+...+Q_{\nu,a}(A_{\ell};\rho)];\;\;\;\ell=1,...,d^2.
\end{eqnarray}
Measurements with the less important for $\rho$ projectors $\Pi(A_1),...,\Pi(A_\ell)$ (using different ensembles of the same quantum state because they do not commute) will give `yes'
with probabilities the $\ell$ smallest $Q$ functions $Q_{\nu,a}(A_1;\rho),..., Q_{\nu,a}(A_\ell;\rho)$.
The ${\cal L}_{\nu,a}(\ell;\rho)$ is an increasing function of $\ell$, and it is the sum of these $\ell$ smallest probabilities (times a normalization factor).

Special cases are
\begin{eqnarray}
{\cal L}_{\nu,a}(1;\rho)=\frac{1}{da}Q_{\nu,a}(A_1;\rho);\;\;\;{\cal L}_{\nu,a}(d^2;\rho)=1.
\end{eqnarray}

\begin{proposition}\label{pro250}
\mbox{}
\begin{itemize}
\item[(1)]
\begin{eqnarray}\label{dd}
{\cal L}_{\nu,a}(\ell;\rho)\le \frac{\ell}{d^2}.
\end{eqnarray}
\item[(2)]
If the $Q$ function ${\mathfrak Q}[\rho, {\mathfrak N}(\nu; a)]$ is finer than ${\mathfrak Q}[\rho, {\mathfrak N}(\mu; b)]$ 
( ${\mathfrak Q}[\rho, {\mathfrak N}(\nu; a)]\overset {Q}\sqsubset {\mathfrak Q}[\rho, {\mathfrak N}(\mu; b)]$) then
\begin{eqnarray}\label{82}
a{\cal L}_{\nu,a}(\ell;\rho)\le b{\cal L}_{\mu,b}(\ell;\rho).
\end{eqnarray}
\end{itemize}
\end{proposition}
\begin{proof}
\mbox{}
\begin{itemize}
\item[(1)]
We start from Eq.(\ref{62}) which we rewrite as
\begin{eqnarray}
\frac{1}{da}[Q_{\nu,a}(A_1;\rho)+Q_{\nu,a}(A_2;\rho)+...+Q_{\nu,a}(A_{d^2};\rho)]=1.
\end{eqnarray}
For $k>\ell$ we have $Q_{\nu,a}(A_k;\rho)\ge Q_{\nu,a}(A_\ell;\rho)$, and we replace the $Q_{\nu,a}(A_k;\rho)$ with the $Q_{\nu,a}(A_\ell;\rho)$. We get
\begin{eqnarray}\label{86}
{\cal L}_{\nu,a}(\ell;\rho)+\frac{d^2-\ell}{da}Q_{\nu,a}(A_\ell;\rho)\le 1.
\end{eqnarray}
At this stage we consider two cases:
\begin{itemize}
\item[(i)]
If  $Q_{\nu,a}(A_\ell;\rho)\le \frac{a}{d}$ we prove Eq.~(\ref{dd}) as follows.
Taking into account that $a<d$ we rewrite Eq.~(\ref{86}) as
\begin{eqnarray}
{\cal L}_{\nu,a}(\ell;\rho)\le \frac{1}{da}[\ell Q_{\nu,a}(A_\ell;\rho)]\le \frac{1}{da}\left (\ell \frac{a}{d}\right )\le  \frac{\ell}{d^2}.
\end{eqnarray}
This proves Eq.~(\ref{dd}) for this case.
\item[(ii)]
If  $Q_{\nu,a}(A_\ell;\rho)> \frac{a}{d}$ we get
\begin{eqnarray}
{\cal L}_{\nu,a}(\ell;\rho)+\frac{d^2-\ell}{da}Q_{\nu,a}(A_\ell;\rho)>{\cal L}_{\nu,a}(\ell;\rho)+\frac{d^2-\ell}{da}\frac{a}{d}={\cal L}_{\nu,a}(\ell;\rho)+1-\frac{\ell}{d^2}
\end{eqnarray}
From this and Eq.~(\ref{86}) follows Eq.~(\ref{dd}).
\end{itemize}
\item[(2)]
If the $Q$ function ${\mathfrak Q}[\rho, {\mathfrak N}(\nu; a)]$ is finer than ${\mathfrak Q}[\rho, {\mathfrak N}(\mu; b)]$,
 there is a bijective map between $Q_{\nu,a}(A_i;\rho)\in {\mathfrak Q}[\rho, {\mathfrak N}(\nu;a)]$ and $Q_{\mu,b}(B_i;\rho)\in {\mathfrak Q}[\rho, {\mathfrak N}(\mu;b)]$  such that $Q_{\nu,a}(A_i;\theta)\le Q_{\mu,b}(B_i;\rho)$.
For the $\ell$ lowest values of $ Q_{\mu,b}(B_i;\rho)$ and we get
\begin{eqnarray}
Q_{\nu,a}(A_1;\rho)+Q_{\nu,a}(A_2;\rho)+...+Q_{\nu,a}(A_{\ell};\rho)\le Q_{\mu,b}(B_1;\rho)+Q_{\mu,b}(B_2;\rho)+...+Q_{\mu,b}(B_{\ell};\rho)
\end{eqnarray}
The right hand side is $b{\cal L}_{\mu,b}(\ell;\rho)$.
The left hand side is greater or equal to $a{\cal L}_{\nu,a}(\ell;\rho)$ (because they might not be the $\ell$ lowest values of $ Q_{\nu, a}(A_i;\rho)$). 
This proves the proposition.
\end{itemize}
\end{proof}

Eq.(\ref{82}) shows that the Lorenz values of a coarser $Q$ function are larger than the corresponding  Lorenz values of a finer $Q$ function, times a factor.

\begin{proposition}\label{pro66}
The Lorenz values are a superadditive function:
 \begin{eqnarray}\label{hhh}
{\cal L}_{\nu,a}[\ell;p\rho+(1-p)\sigma]\ge p{\cal L}_{\nu,a}(\ell;\rho)+(1-p){\cal L}_{\nu,a}(\ell;\sigma);\;\;\;0\le p\le 1.
\end{eqnarray}
If $\rho, \sigma$ are comonotonic density matrices then the following additivity property holds:
\begin{eqnarray}\label{QA1}
{\cal L}_{\nu,a}[\ell;p\rho+(1-p)\sigma]=p{\cal L}_{\nu,a}(\ell;\rho)+(1-p){\cal L}_{\nu,a}(\ell;\sigma);\;\;\;0\le p\le 1.
\end{eqnarray}
\end{proposition}
\begin{proof}
We have
\begin{eqnarray}\label{V1}
&&{\cal L}_{\nu,a}[\ell;p\rho+(1-p)\sigma]=\frac{1}{da}\left \{Q_{\nu,a}[A_1;p\rho+(1-p)\sigma]+...+Q_{\nu,a}[A_{\ell};p\rho+(1-p)\sigma]\right \}
\nonumber\\&&=\frac{p}{da}
[Q_{\nu,a}(A_1;\rho)+...+Q_{\nu,a}(A_{\ell};\rho)]+\frac{1-p}{da}[Q_{\nu,a}(A_1;\sigma)+...+Q_{\nu,a}(A_{\ell};\sigma)]
\end{eqnarray}
$A_1,...,A_{d^2}$ is the ranking permutation for $p\rho+(1-p)\sigma$, and in general it will not be the 
ranking permutation for $\rho$ and $\sigma$. Therefore
\begin{eqnarray}\label{V2}
&&\frac{p}{da}[Q_{\nu,a}(A_1;\rho)+...+Q_{\nu,a}(A_{\ell};\rho)]\ge p{\cal L}_{\nu,a}(\ell;\rho)\nonumber\\
&&\frac{1-p}{da}[Q_{\nu,a}(A_1;\sigma)+...+Q_{\nu,a}(A_{\ell};\sigma)]\ge (1-p){\cal L}_{\nu,a}(\ell;\sigma)
\end{eqnarray}
Combining Eqs~(\ref{V1}), (\ref{V2}) we prove Eq.~(\ref{hhh}).

For comonotonic density matrices the ranking permutation of $p\rho+(1-p)\sigma$ is the same as the ranking permutation of $\rho, \sigma$.
In this case in Eq.(\ref{V2}) we have equalities and this proves Eq.~(\ref{QA1}).
\end{proof}
\begin{example}\label{ex11}
We consider the case $\rho =\frac{1}{d}{\bf 1}$. 
Any permutation is ranking permutation of this density matrix.
Then
 \begin{eqnarray}
Q_{\nu,a}\left (A;\frac{1}{d}{\bf 1}\right )=\frac{a}{d};\;\;\;P_{\nu,a}\left (A;\frac{1}{d}{\bf 1}\right )=\frac{1}{d^2a};\;\;\;E_{\nu, a}=\log (d^2);\;\;\;{\cal L}_{\nu,a}\left (\ell;\frac{1}{d}{\bf 1}\right )=\frac{\ell}{d^2}
\end{eqnarray}
where $A\in{\mathfrak N}(\nu,a)$.
This state is uniformly distributed  in the granular structure of the phase space that consists of coherent subspaces in ${\mathfrak H}(\nu, a)$ (all the $Q_{\nu,a}\left (A_i;\frac{1}{d}{\bf 1}\right )$ are equal to each other).

In the case $a=2$ we get
\begin{eqnarray}
Q_{\nu,2}\left (\{i,j\};\frac{1}{d}{\bf 1}\right )=\frac{2}{d};\;\;\;\delta Q_{\nu,2}\left (\{i,j\};\frac{1}{d}{\bf 1}\right )=0.
\end{eqnarray}
The $\rho =\frac{1}{d}{\bf 1}$ can be viewed as a `semi-classical' state in the sense that  $\delta Q_{\nu,2}\left (\{i,j\};\frac{1}{d}{\bf 1}\right )=0$.
\end{example}

\subsection{Lorenz operators: truncated resolutions of the identity}

We introduce the Lorenz operators associated with the permutation $\pi \in {\cal S}_{d^2}$, as
\begin{eqnarray}\label{100}
\Lambda_{\nu,a,\pi}(\ell)=\frac{1}{da}\left \{\Pi[\pi(B_1)]+\Pi[\pi(B_2)]+...+\Pi[\pi(B_{\ell})]\right \};\;\;\;\ell=1,...,d^2.
\end{eqnarray}
Here $(B_1,...,B_{d^2})$ is the reference ordering of the subsets in ${\mathfrak N}(\nu;a)$.
The Lorenz operators  are positive semi-definite operators (as sum of projectors) and 
\begin{eqnarray}
{\rm Tr}[\Lambda_{\nu,a,\pi}(\ell)]=\frac{\ell}{d};\;\;\;\Lambda_{\nu,a,\pi}(d^2)={\bf 1}.
\end{eqnarray}
From Eqs.~(\ref{2A}),(\ref{100}) we see that the Lorenz operators are truncated resolutions of the identity, and as the permutation $\pi$ changes we get various orderings.

If the ranking permutation of the density matrix $\rho$ is $\pi$, then 
\begin{eqnarray}\label{83}
{\cal L}_{\nu,a}(\ell;\rho)={\rm Tr}[\Lambda_{\nu,a,\pi}(\ell)\rho].
\end{eqnarray}
We note that for the borderline   density matrices in $R\setminus {\mathfrak R}$, we can use any of the ranking permutations 
of $\rho$ (they will give the same result).

\section{The Gini index: Inequality in the distribution of the $Q$ function}

Example \ref{ex11} above shows that for $\rho =\frac{1}{d}{\bf 1}$ we get a uniform distribution of the $Q$-function related to the coherent subspaces within an equivalence class, and this gives  ${\cal L}_{\nu,a}\left (\ell;\frac{1}{d}{\bf 1}\right )=\frac{\ell}{d^2}$.
The quantity
\begin{eqnarray}\label{84}
{\cal G}_{\nu,a}(\rho)=\frac{1}{\cal N}\sum _{\ell=1}^{d^2}\left[{\cal L}_{\nu,a}\left (\ell;\frac{1}{d}{\bf 1}\right )-{\cal L}_{\nu,a}(\ell;\rho)\right ]=\frac{1}{\cal N}\sum _{\ell=1}^{d^2}\left[\frac{\ell}{d^2}-{\cal L}_{\nu,a}(\ell;\rho)\right ];\;\;\;
{\cal N}=\sum _{\ell=1}^{d^2}\frac{\ell}{d^2}=\frac{d^2+1}{2}
\end{eqnarray}
is analogous to the Gini index in Mathematical Economics, and in the present context measures how close is the $Q$ function to a uniform distribution in phase space.
It is easily seen that
\begin{eqnarray}\label{85}
{\cal G}_{\nu,a}(\rho)&=&1-\frac{2}{d^2+1}\sum _{\ell=1}^{d^2}{\cal L}_{\nu,a}(\ell;\rho)=\frac{d^2-1}{d^2+1}-\frac{2}{d^2+1}\sum _{\ell=1}^{d^2-1}{\cal L}_{\nu,a}(\ell;\rho)\nonumber\\&=&
1-\frac{2}{(d^2+1)da}[d^2Q_{\nu,a}(A_1;\rho)+(d^2-1)Q_{\nu,a}(A_2;\rho)+...+Q_{\nu,a}(A_{d^2};\rho)]
\end{eqnarray}
We see that bounds for the Gini index are 
\begin{eqnarray}
0\le {\cal G}_{\nu,a}(\rho)< \frac{d^2-1}{d^2+1}.
\end{eqnarray}
By definition the Gini index  for $\rho=\frac{1}{d}{\bf 1}$ is ${\cal G}_{\nu,a}\left (\frac{1}{d}{\bf 1}\right )=0$.

The Gini index shows the inequality in the distribution of the $Q$ function of a quantum state, in the granular structure of the Hilbert space that consists of coherent subspaces in ${\mathfrak H}(\nu, a)$.
It is close to $0$ for states which are almost uniformly distributed in all of these coherent subspaces, and it is close to $\frac{d^2-1}{d^2+1}$ for states which `live' primarily in a few of these coherent subspaces.
\begin{remark}
It is an open problem to find the maximum value (smaller than $\frac{d^2-1}{d^2+1}$) that ${\cal G}_{\nu,a}(\rho)$ can take in the set of all density matrices.
\end{remark}

\begin{proposition}\label{pro251}
If the $Q$ function ${\mathfrak Q}[\rho, {\mathfrak N}(\nu; a)]$ is finer than ${\mathfrak Q}[\rho, {\mathfrak N}(\mu; b)]$ 
( ${\mathfrak Q}[\rho, {\mathfrak N}(\nu; a)]\overset {Q}\sqsubset {\mathfrak Q}[\rho, {\mathfrak N}(\mu; b)]$) then
\begin{eqnarray}\label{97}
a[1-{\cal G}_{\nu,a}(\rho)]\le b[1-{\cal G}_{\nu,b}(\rho)].
\end{eqnarray}
\end{proposition}
\begin{proof}
From Eq.(\ref{82}) it follows that
\begin{eqnarray}\label{85}
a\sum _{\ell=1}^{d^2}{\cal L}_{\nu,a}(\ell;\rho)\le b\sum _{\ell=1}^{d^2}{\cal L}_{\mu,b}(\ell;\rho)
\end{eqnarray}
This can be written as
\begin{eqnarray}
\frac {d^2+1}{2}a[1-{\cal G}_{\nu,a}(\rho)]\le \frac {d^2+1}{2}b[1-{\cal G}_{\nu,b}(\rho)],
\end{eqnarray}
and from this follows Eq.(\ref{97}).
\end{proof}
We call  complementary Gini index the $1-{\cal G}_{\nu,a}(\rho)$.
Eq.(\ref{97}) shows that the complementary Gini index of a coarser $Q$ function is larger than the complementary Gini index of a finer $Q$ function, times a factor.

\begin{proposition}
A quantum state $\rho$ and its displacement $\rho^{\prime} =D(i)\rho [D(i)]^\dagger$ have the same Lorenz values, the same Gini index, and the same entropy of Eq.(\ref{ent}).
\end{proposition}
\begin{proof}
We have seen in Eq.~(\ref{Q12}) that 
 \begin{eqnarray}
Q_{\nu,a}(A;\rho)=Q_{\nu,a}(A^{\prime};\rho^{\prime});\;\;\;A^{\prime}=A+i.
\end{eqnarray}
Therefore when we order the $Q$ function for $\rho$ as
\begin{eqnarray}
Q_{\nu,a}(A_1;\rho)\le Q_{\nu,a}(A_2;\rho)\le ...\le Q_{\nu,a}(A_{d^2};\rho),
\end{eqnarray}
the corresponding ordering of the $Q$ function for $\rho^{\prime}$ is
\begin{eqnarray}
Q_{\nu,a}(A^{\prime}_1;\rho^{\prime})\le Q_{\nu,a}(A^{\prime}_2;\rho^{\prime})\le ...\le Q_{\nu,a}(A^{\prime}_{d^2};\rho^{\prime});\;\;\;A^{\prime}_r=A_r+i.
\end{eqnarray}
Consequently
\begin{eqnarray}
{\cal L}_{\nu,a}(\ell;\rho)&=&\frac{1}{da}[Q_{\nu,a}(A_1;\rho)+Q_{\nu,a}(A_2;\rho)+...+Q_{\nu,a}(A_{\ell};\rho)]\nonumber\\&=&
\frac{1}{da}[Q_{\nu,a}(A^{\prime}_1;\rho^{\prime})+Q_{\nu,a}(A^{\prime}_2;\rho^{\prime})+...+Q_{\nu,a}(A^{\prime}_{\ell};\rho^{\prime})]={\cal L}_{\nu,a}(\ell;\rho^{\prime})
\end{eqnarray}
We have proved that the states $\rho, \rho ^{\prime}$ have the same Lorenz values and therefore they have the same Gini index.
From Eq.~(\ref{Q12}) we easily see that they also have the same entropy  of Eq.(\ref{ent}).

\end{proof}

\begin{proposition}\label{pro67}
The Gini index is a subadditive function:
\begin{eqnarray}
{\cal G}_{\nu,a}[p\rho+(1-p)\sigma]\le p{\cal G}_{\nu,a}(\rho)+(1-p){\cal G}_{\nu,a}(\sigma);\;\;\;0\le p\le 1.
\end{eqnarray}
If $\rho, \sigma$ are comonotonic density matrices then the following additivity property holds:
\begin{eqnarray}\label{QA2}
{\cal G}_{\nu,a}[p\rho+(1-p)\sigma]=p{\cal G}_{\nu,a}(\rho)+(1-p){\cal G}_{\nu,a}(\sigma);\;\;\;0\le p\le 1.
\end{eqnarray}
\end{proposition}
\begin{proof}
This follows easily from proposition \ref{pro66}.
\end{proof}

\subsection{Gini operators}
We introduce the Gini operators associated with the equivalence class of density matrices ${\mathfrak R}(\nu,a,\pi)$, as
\begin{eqnarray}
{\mathfrak G}_{\nu,a,\pi}=\frac{2}{d^2+1}\sum _{\ell=1}^{d^2}\Lambda_{\nu,a,\pi}(\ell)=
\frac{2}{(d^2+1)da}[d^2\Pi(A_1)+(d^2-1)\Pi(A_2)+...+\Pi(A_{d^2})]
\end{eqnarray}
The ordering of the projectors is defined by $\pi$ (with $A_i=\pi(B_i)$).
The Gini operators change abruptly as $\pi$ changes when we move from one equivalence class to another.

The Gini operators  are positive semi-definite operators and 
\begin{eqnarray}\label{830}
{\cal G}_{\nu,a}(\rho)=1-{\rm Tr}(\rho{\mathfrak G}_{\nu,a, \pi});\;\;\;{\rm Tr}[{\mathfrak G}_{\nu,a,\pi}]=d.
\end{eqnarray}
Here $\pi$ is the ranking permutation of the density matrix $\rho$.
We note that for the borderline   density matrices in $R\setminus {\mathfrak R}$, we can use any of the ranking permutations 
of $\rho$ (they will give the same result). 

We have explained earlier that the values of the $Q$ function are measurable (using different ensembles of the same density matrix $\rho$, because the projectors do not commute)
and in this sense both the Lorenz values and the Gini index are measurable quantities.

\subsection{Examples}

We consider a quantum system with variables in  ${\mathbb Z}_3$ and the orthonormal basis of position states
 \begin{eqnarray}
\ket{X;0}=\begin{pmatrix}
1\\
0\\
0\\
\end{pmatrix};\;\;\;
\ket{X;1}=\begin{pmatrix}
0\\
1\\
0\\
\end{pmatrix};\;\;\;
\ket{X;2}=\begin{pmatrix}
0\\
0\\
1\\
\end{pmatrix}.
\end{eqnarray}
For coherent states we will use the fiducial vector
 \begin{eqnarray}
\ket{s}=\frac{1}{\sqrt {14}}\begin{pmatrix}
1\\
2\\
3\\
\end{pmatrix}.
\end{eqnarray}
Let $\rho$ be the density matrix
\begin{eqnarray}\label{dm}
\rho=\frac{1}{3}\begin{pmatrix}
1&1&0\\
1&1&0\\
0&0&1\\
\end{pmatrix}
\end{eqnarray}
which describes a mixed state.
We have  calculated the `standard' $Q$ and $P$ functions (i.e., for all subsets of the phase space with cardinality $1$). Numerical results are shown in table \ref{t1}.
The entropy of Eq.~(\ref{ent}) is in this case $E_{1,1}=2.1195$.

We also calculated the $Q_{\nu,2}(\{i,j\};\rho)$ of Eq.(\ref{Q}), the $\delta Q_{\nu,2}(\{i,j\};\rho)$ of Eq.(\ref{456}), and the $P_{\nu,2}(\{i,j\};\rho)$ of Eq.(\ref{ff}).
Results for the  ${\mathfrak n}(2)/9=4$
equivalence classes ${\mathfrak N}(\nu,2)$ are shown in tables \ref{t2}, \ref{t3}, \ref{t4}, \ref{t5}.
As we explained earlier the  $\delta Q_{\nu,2}(\{i,j\};\rho)$  quantify quantum interference, and their average value within an equivalence class is zero (Eq.(\ref{67})).
The entropies of Eq.~(\ref{ent}) are in these cases 
\begin{eqnarray}
E_{1,2}=2.1790;\;\;\;E_{2,2}=2.1785;\;\;\;E_{3,2}=2.1724;\;\;\;E_{4,2}=2.1724, 
\end{eqnarray}
correspondingly.

Then we ordered the sets in each of the equivalence classes ${\mathfrak N}(1,1)$, ${\mathfrak N}(1,2)$, ${\mathfrak N}(2,2)$, ${\mathfrak N}(3,2)$, ${\mathfrak N}(4,2)$, so that the corresponding $Q$-function is in ascending order as in Eq.~(\ref{order}). 
The corresponding  Lorenz functions are shown in table \ref{t6}, and the corresponding Gini indices are
\begin{eqnarray}
{\cal G}_{1,1}(\rho)=0.1952;\;\;\;{\cal G}_{1,2}(\rho)=0.1920;\;\;\;{\cal G}_{2,2}(\rho)=0.0838;\;\;\;{\cal G}_{3,2}(\rho)=0.1140;\;\;\;{\cal G}_{4,2}(\rho)=0.1140.
\end{eqnarray}
It is easily seen that the inequalities of Eqs.(\ref{82}),(\ref{97}) hold.
\begin{remark}
Table \ref{t6} shows that for the particular example that we considered we have ${\cal L}_{3,2}(\ell;\rho)={\cal L}_{4,2}(\ell;\rho)$, but this is not a general result.
\end{remark}

\begin{table}
\caption{The `standard' $Q$ and $P$ functions  (i.e., for sets in the equivalence class ${\mathfrak N}(1,1)$)  for the  density matrix in Eq.(\ref{dm}). The entropy of Eq.(\ref{ent}) is $E_{1,1}=2.1195$.}
\def\arraystretch{2}
\begin{tabular}{|c|c|c|c|c|c|c|c|c|c|}\hline
$i$&$(0,0)$&$(0,1)$&$(0,2)$&$(1,0)$&$(1,1)$&$(1,2)$&$(2,0)$&$(2,1)$&$(2,2)$\\\hline
$Q_{1,1}(i;\rho)$&$0.4286$&$0.4762$&$0.6190$&$0.2857$&$0.2619$&$0.1905$&$0.2857$&$0.2619$&$0.1905$\\\hline
$P_{1,1}(i;\rho)$&$-0.1935$&$0.0458$&$0.7638$&$0.2634$&$0.1437$&$-0.2152$&$0.2634$&$0.1437$&$-0.2152$\\\hline
\end{tabular} \label{t1}
\end{table}
\begin{table}
\caption{The equivalence class  ${\mathfrak N}(1,2)$ of subsets of the phase space ${\mathbb Z}_3\times {\mathbb Z}_3$. 
For each pair of points the values of $Q_{1,2}(\{i,j\};\rho), \delta Q_{1,2}(\{i,j\};\rho), P_{1,2}(\{i,j\};\rho)$ for the density matrix in Eq.(\ref{dm}), are  shown.  The entropy of Eq.(\ref{ent}) is $E_{1,2}=2.1790$.}
\def\arraystretch{2}
\begin{tabular}{|c|c|c|c|}\hline
${\mathfrak N}(1;2)$&$Q_{1,2}(\{i,j\};\rho)$&$\delta Q_{1,2}(\{i,j\};\rho)$&$P_{1,2}(\{i,j\};\rho)$\\\hline
$\{(0,0),(0,1)\}$&$0.3467$&$-0.5581$&$0.0102$\\\hline
$\{(0,1),(0,2)\}$&$0.9867$&$-0.1085$&$0.1527$\\\hline
$\{(0,2),(0,0)\}$&$0.6667$&$-0.3809$&$0.5088$\\\hline
$\{(1,0),(1,1)\}$&$0.3467$&$-0.2009$&$0.0782$\\\hline
$\{(1,1),(1,2)\}$&$0.9867$&$0.5343$&$0.0070$\\\hline
$\{(1,2),(1,0)\}$&$0.6667$&$0.1905$&$-0.1711$\\\hline
$\{(2,0),(2,1)\}$&$0.3467$&$-0.2009$&$0.0782$\\\hline
$\{(2,1),(2,2)\}$&$0.9867$&$0.5343$&$0.0070$\\\hline
$\{(2,2),(2,0)\}$&$0.6667$&$0.1905$&$-0.1711$\\\hline
\end{tabular} \label{t2}
\end{table}
\begin{table}
\caption{The equivalence class  ${\mathfrak N}(2,2)$ of subsets of the phase space ${\mathbb Z}_3\times {\mathbb Z}_3$. 
For each pair of points the values of $Q_{2,2}(\{i,j\};\rho), \delta Q_{2,2}(\{i,j\};\rho), P_{2,2}(\{i,j\};\rho)$ for the  density matrix in Eq.~(\ref{dm}), are shown.  The entropy of Eq.~(\ref{ent}) is $E_{2,2}=2.1785$.}
\def\arraystretch{2}
\begin{tabular}{|c|c|c|c|}\hline
${\mathfrak N}(2;2)$&$Q_{2,2}(\{i,j\};\rho)$&$\delta Q_{2,2}(\{i,j\};\rho)$&$P_{2,2}(\{i,j\};\rho)$ \\\hline
$\{(0,0),(1,0)\}$&$0.7891$&$0.0748$&$0.4084$\\\hline
$\{(0,1),(1,1)\}$&$0.7843$&$0.0462$&$0.1060$\\\hline
$\{(0,2),(1,2)\}$&$0.7075$&$-0.1020$&$-0.1965$\\\hline
$\{(1,0),(2,0)\}$&$0.4218$&$-0.1496$&$-0.6502$\\\hline
$\{(1,1),(2,1)\}$&$0.5034$&$-0.0204$&$-0.0453$\\\hline
$\{(1,2),(2,2)\}$&$0.5850$&$0.2040$&$0.5597$\\\hline
$\{(2,0),(0,0)\}$&$0.7891$&$0.0748$&$0.4084$\\\hline
$\{(2,1),(0,1)\}$&$0.7843$&$0.0462$&$0.1060$\\\hline
$\{(2,2),(0,2)\}$&$0.7075$&$-0.1020$&$-0.1965$\\\hline
\end{tabular} \label{t3}
\end{table}
\begin{table}
\caption{The equivalence class  ${\mathfrak N}(3,2)$ of subsets of the phase space ${\mathbb Z}_3\times {\mathbb Z}_3$. 
For each pair of points the values of $Q_{3,2}(\{i,j\};\rho), \delta Q_{3,2}(\{i,j\};\rho), P_{3,2}(\{i,j\};\rho)$ for the  density matrix in Eq.~(\ref{dm}), are shown.  The entropy of Eq.~(\ref{ent}) is $E_{3,2}=2.1724$.}
\def\arraystretch{2}
\begin{tabular}{|c|c|c|c|}\hline
 ${\mathfrak N}(3;2)$&$Q_{3,2}(\{i,j\};\rho)$&$\delta Q_{3,2}(\{i,j\};\rho)$&$P_{3,2}(\{i,j\};\rho)$\\\hline
$\{(0,0),(1,2)\}$&$0.6138$&$-0.0053$&$-0.3139$\\\hline
$\{(0,1),(1,0)\}$&$0.8270$&$0.0651$&$0.2143$\\\hline
$\{(0,2),(1,1)\}$&$0.7887$&$-0.0922$&$0.3892$\\\hline
$\{(1,0),(2,2)\}$&$0.4882$&$0.0120$&$-0.0615$\\\hline
$\{(1,1),(2,0)\}$&$0.4499$&$-0.0977$&$0.1134$\\\hline
$\{(1,2),(2,1)\}$&$0.5373$&$0.0849$&$-0.1661$\\\hline
$\{(2,0),(0,2)\}$&$0.8980$&$-0.0067$&$0.5421$\\\hline
$\{(2,1),(0,0)\}$&$0.7231$&$0.0326$&$-0.1610$\\\hline
$\{(2,2),(0,1)\}$&$0.6740$&$0.0073$&$-0.0564$\\\hline
\end{tabular} \label{t4}
\end{table}
\begin{table}
\caption{The equivalence class  ${\mathfrak N}(4,2)$ of subsets of the phase space ${\mathbb Z}_3\times {\mathbb Z}_3$. 
For each pair of points the values of $Q_{4,2}(\{i,j\};\rho), \delta Q_{4,2}(\{i,j\};\rho), P_{4,2}(\{i,j\};\rho)$ for the  density matrix in Eq.~(\ref{dm}), are shown.  The entropy of Eq.~(\ref{ent}) is $E_{4,2}=2.1724$.}
\def\arraystretch{2}
\begin{tabular}{|c|c|c|c|}\hline
${\mathfrak N}(4;2)$&$Q_{4,2}(\{i,j\};\rho)$&$\delta Q_{4,2}(\{i,j\};\rho)$&$P_{4,2}(\{i,j\};\rho)$\\\hline
$\{(0,0),(1,1)\}$&$0.7231$&$0.0326$&$-0.1610$\\\hline
$\{(0,1),(1,2)\}$&$0.6740$&$0.0073$&$-0.0564$\\\hline
$\{(0,2),(1,0)\}$&$0.8980$&$-0.0067$&$0.5421$\\\hline
$\{(1,0),(2,1)\}$&$0.4499$&$-0.0977$&$0.1134$\\\hline
$\{(1,1),(2,2)\}$&$0.5373$&$0.0849$&$-0.1661$\\\hline
$\{(1,2),(2,0)\}$&$0.4882$&$0.0120$&$-0.0615$\\\hline
$\{(2,0),(0,1)\}$&$0.8270$&$0.0651$&$0.2143$\\\hline
$\{(2,1),(0,2)\}$&$0.7887$&$-0.0922$&$0.3892$\\\hline
$\{(2,2),(0,0)\}$&$0.6138$&$-0.0053$&$-0.3139$\\\hline
\end{tabular} \label{t5}
\end{table}
\begin{table}
\caption{
Ordering of  the sets in each of the equivalence classes ${\mathfrak N}(1,1)$, ${\mathfrak N}(1,2)$, ${\mathfrak N}(2,2)$, ${\mathfrak N}(3,2)$, ${\mathfrak N}(4,2)$, so that the corresponding $Q$-function is in ascending order as in Eq.(\ref{order}). 
The corresponding  Lorenz functions are shown.}
\def\arraystretch{2}
\begin{tabular}{|c|c||c|c||c|c||c|c||c|c|}\hline
 \multicolumn{2}{|c||}{ ${\mathfrak N}(1,1)$}& \multicolumn{2}{c||}{ ${\mathfrak N}(1,2)$}& \multicolumn{2}{c||}{ ${\mathfrak N}(2,2)$}& \multicolumn{2}{c||}{ ${\mathfrak N}(3,2)$}& \multicolumn{2}{c|}{ ${\mathfrak N}(4,2)$}\\\hline
$i$&${\cal L}_{1,1}(\ell ;\rho)$&$\{i,j\}$&${\cal L}_{1,2}(\ell ;\rho)$&$\{i,j\}$&${\cal L}_{2,2}(\ell ;\rho)$&$\{i,j\}$&${\cal L}_{3,2}(\ell ;\rho)$&$\{i,j\}$&${\cal L}_{4,2}(\ell ;\rho)$\\\hline
$(1,2)$&$0.0635$&\{(0,0),(0,1)\}&$0.0578$&\{(1,0),(2,0)\}&$0.0703$&\{(1,1),(2,0)\}&$0.0750$&\{(1,0),(2,1)\}&$0.0750$\\\hline
$(2,2)$&$0.1270$&\{(1,0),(1,1)\}&$0.1155$&\{(1,1),(2,1)\}&$0.1542$&\{(1,0),(2,2)\}&$0.1563$&\{(1,2),(2,0)\}&$0.1563$\\\hline
$(1,1)$&$0.2143$&\{(2,0),(2,1)\}&$0.1733$&\{(1,2),(2,2)\}&$0.2517$&\{(1,2),(2,1)\}&$0.2459$&\{(1,1),(2,2)\}&$0.2459$\\\hline
$(2,1)$&$0.3016$&\{(0,2),(0,0)\}&$0.2844$&\{(0,2),(1,2)\}&$0.3696$&\{(0,0),(1,2)\}&$0.3482$&\{(2,2),(0,0)\}&$0.3482$\\\hline
$(1,0)$&$0.3968$&\{(1,2),(1,0)\}&$0.3956$&\{(2,2),(0,2)\}&$0.4875$&\{(2,2),(0,1)\}&$0.4605$&\{(0,1),(1,2)\}&$0.4605$\\\hline
$(2,0)$&$0.4921$&\{(2,2),(2,0)\}&$0.5067$&\{(0,1),(1,1)\}&$0.6182$&\{(2,1),(0,0)\}&$0.5810$&\{(0,0),(1,1)\}&$0.5810$\\\hline
$(0,0)$&$0.6349$&\{(0,1),(0,2)\}&$0.6711$&\{(2,1),(0,1)\}&$0.7489$&\{(0,2),(1,1)\}&$0.7125$&\{(2,1),(0,2)\}&$0.7125$\\\hline
$(0,1)$&$0.7937$&\{(1,1),(1,2)\}&$0.8356$&\{(0,0),(1,0)\}&$0.8804$&\{(0,1),(1,0)\}&$0.8503$&\{(2,0),(0,1)\}&$0.8503$\\\hline
$(0,2)$&$1$&\{(2,1),(2,2)\}&$1$&\{(2,0),(0,0)\}&$1$&\{(2,0),(0,2)\}&$1$&\{(0,2),(1,0)\}&$1$\\\hline
\end{tabular} \label{t6}
\end{table}

\section{Summary}

We have considered quantum systems with variables in ${\mathbb Z}_d$, where $d$ is a prime odd number. In this general context, we considered multi-dimensional coherent subspaces. We have shown that
the set of all coherent subspaces is partitioned into equivalence classes, with $d^2$ subspaces in each class.
The coherent projectors in a given  equivalence class have the closure property under displacements, and the resolution of the identity. 

The projectors related to two-dimensional coherent subspaces obey the `non-additivity' inequality in Eq.~(\ref{dfg}).
This  led to `non-additivity operators' which are related to the non-commutativity of the projectors and quantify quantum interference in phase space.

Based on coherent projectors, we introduced generalized $Q$ and $P$ functions of density matrices.
The ordering of the values of the $Q$ function led to the ranking permutation of a density matrix.
Comonotonic density matrices have the same ranking permutation.

The Lorenz values and the Gini index provide a novel approach to phase space methods in quantum mechanics.
They quantify the inequality in the distribution of the $Q$ function of a quantum state.
The Lorenz values are a superadditive function and the Gini index a subadditive function (propositions \ref{pro66}, \ref{pro67}). 
For comonotonic density matrices they are both  additive quantities.
A comparison of  the Lorenz values and the Gini index in the cases of coarse and fine coverings of the Hilbert space, is made in propositions \ref{pro250}, \ref{pro251}.
Various examples have been given that demonstrate these ideas. 
 
The work brings novel concepts in  the area of phase space methods, firstly by using multi-dimensional coherent subspaces and coherent projectors, and secondly by using 
the Lorenz values and the Gini index to quantify inequalities in the distribution of the $Q$ function of a quantum state.

\end{document}